\documentclass[journal,comsoc]{IEEEtran}
\usepackage{algorithm,amsbsy,amsmath,amssymb,epsfig,bbm,mathrsfs,multirow,amsthm}
\usepackage{algpseudocode}
\usepackage{float}
\usepackage{graphicx,caption,subcaption,array,multirow}%for three figures side-by-side
\usepackage{makecell}
\usepackage{color}
\usepackage[hidelinks]{hyperref}
\DeclareMathAlphabet{\mathpzc}{OT1}{pzc}{m}{it}
\DeclareMathOperator*{\argmax}{\arg\!\max}

\newtheorem{lemma}{Lemma}
\newtheorem{theorem}{\textbf{\textsc{Theorem}}}
\newtheorem{cor}{Corollary}

\begin{document}
	
\title{Optimal and Fast Real-time Resources Slicing with Deep Dueling Neural Networks}

\author{Nguyen Van Huynh, Dinh Thai Hoang, Diep N. Nguyen, and Eryk Dutkiewicz \\
\thanks{This work was supported in part by the National Research Foundation of Korea (NRF) grant funded by the Korean government (MSIP) (2014R1A5A1011478). This paper was presented in part at IEEE International Conference  on Communications (IEEE ICC'2014), Sydney, Australia~\cite{niyato20xx}.}
\thanks{D. T. Hoang, D. Niyato, and P. Wang are with Nanyang Technological University, Singapore. E-mails: dinh0007@e.ntu.edu.sg, dniyato@ntu.edu.sg, wangping@ntu.edu.sg.}
}

\author{Nguyen Van Huynh,~\IEEEmembership{Student Member,~IEEE,}
	Dinh Thai Hoang,~\IEEEmembership{Member,~IEEE,}
	Diep N. Nguyen,~\IEEEmembership{Member,~IEEE,} \\
	and~Eryk Dutkiewicz,~\IEEEmembership{Senior Member,~IEEE}
%	\thanks{Dinh Thai Hoang is the corresponding author.}
%	\thanks{This paper was presented in part at IEEE International Conference  on Communications (IEEE ICC'2019), Shanghai, China~\cite{ICC}.}
	\thanks{Nguyen Van Huynh, Dinh Thai Hoang, Diep N. Nguyen, and Eryk Dutkiewicz are with University of Technology Sydney, Australia. E-mails: huynh.nguyenvan@student.uts.edu.au, \{Hoang.Dinh, Diep.Nguyen, and Eryk.Dutkiewicz\}@uts.edu.au.} 
}

\maketitle
\begin{abstract}
Effective network slicing requires an infrastructure/network provider to deal with the uncertain demand and real-time dynamics of network resource requests. Another challenge is the combinatorial optimization of numerous resources, e.g., radio, computing, and storage. This article develops an optimal and fast real-time resource slicing framework that maximizes the long-term return of the network provider while taking into account the uncertainty of resource demand from tenants. Specifically, we first propose a novel system model which enables the network provider to effectively slice various types of resources to different classes of users under separate virtual slices. We then capture the real-time arrival of slice requests by a semi-Markov decision process. To obtain the optimal resource allocation policy under the dynamics of slicing requests, e.g., uncertain service time and resource demands, a Q-learning algorithm is often adopted in the literature. However, such an algorithm is notorious for its slow convergence, especially for problems with large state/action spaces. This makes Q-learning practically inapplicable to our case in which multiple resources are simultaneously optimized. To tackle it, we propose a novel network slicing approach with an advanced deep learning architecture, called deep dueling that attains the optimal average reward much faster than the conventional Q-learning algorithm. This property is especially desirable to cope with real-time resource requests and the dynamic demands of users. Extensive simulations show that the proposed framework yields up to 40\% higher long-term average return while being few thousand times faster, compared with state of the art network slicing approaches.
\end{abstract}


%====================================================================================
%====================================================================================
\section{Introduction}

The latest Cisco Visual Networking Index forecast a seven-fold increase in global mobile data traffic from 2016 to 2021, with 5G traffic expected to start having a relatively small but measurable impact on mobile growth starting in 2020. Machine-to-machine (M2M) connections will represent 29\% (3.3 billion) of total mobile connections - up from 5\% (780 million) in 2016~\cite{disruptive_asia_news}. M2M has been the fastest growing mobile connection type as global IoT applications continue to gain traction in consumer and business environments. However, legacy mobile networks are mostly designed to provide services for mobile broadband users and are unable to meet adjustable parameters like priority and quality of service (QoS) for emerging services. Therefore, mobile operators may find difficulties in getting deeply into these emerging vertical services with different service requirements for network design and development. In order to enhance operators' products for vertical enterprises and provide service customization for emerging massive connections, as well as to give more control to enterprises and mobile virtual network operators, the concept of network slicing has been recently introduced to allow the independent usage of a part of network resources by a group of mobile terminals with special requirements. 

Network slicing was introduced by Next Generation Mobile Networks Alliance~\cite{whitepaper}, and it has quickly received a lot of attention from both academia and industry. In general, network slicing is a novel virtualization paradigm that enables multiple logical networks, i.e., slices, to be created according to specific technical or commercial demands and simultaneously run on top of the physical network infrastructure. The core idea of the network slicing is using software-defined networking (SDN) and network functions virtualization (NFV) technologies for virtualizing the physical infrastructure and controlling network operations. In particular, SDN provides a separation between the network control and data planes, improving the flexibility of network function management and efficiency of data transfer. Meanwhile, NFV allows various network functions to be virtualized, i.e., in virtual machines. As a result, the functions can be moved to different locations, and the corresponding virtual machines can be migrated to run on commoditized hardware dynamically depending on the demand and requirements~\cite{Towards5G},~\cite{Zhang2017Network}.

The key benefit of network slicing is to enable providers to offer network services on an as-a-service basis which enhances operational efficiency while reducing time-to-market for new services~\cite{ZhouNetwork2016}. However, to achieve this goal, there are two main challenges in managing network resources. First, many emerging services require not only radio resources for communications but also computing and storage resources to meet requirements about quality-of-service (QoS) of users. For example, IoT services often require simultaneously radio, computing, and storage resources to transmit data and pre-processing intensive tasks because of the resource constraints on IoT devices. Additionally, the network provider may posses multiple data centers with several servers containing diverse resources, e.g., computing and storage, connected together~\cite{Kurtz2018Network}. Thus, how to concurrently manage multiple interconnected resources is an emerging challenge for the network provider. Second, due to the dynamic demand of services, e.g., the frequency of requests and occupation time, and the limitation of resources, how to dynamically allocate resources in a real-time manner to maximize the long term revenue is another challenge of the network provider. As a result, there are some recent works proposing solutions to address these issues.

%======================================================
%======================================================
\subsection{Related Work}

A number of research works have been introduced recently to address the network slicing resource allocation problem for the network provider~\cite{Jiang2016Network}-\cite{ToshibaPatent}. In particular, the authors in~\cite{Jiang2016Network} and~\cite{Soliman2016QoS} developed a two-tier admission control and resource allocation model to answer two fundamental questions, i.e., whether a slice request is accepted and how much radio resource is allocated to the accepted slice. To address this problem, the authors in~\cite{Jiang2016Network} used an extensive searching method to achieve the globally optimal resource allocation solution for the network provider. However, this searching method cannot be applied to complex systems with a large number of resources. To address this problem, a heuristic scheme with three main steps was introduced in~\cite{Soliman2016QoS} to effectively allocate resources to the users. Yet this heuristic scheme cannot guarantee to achieve the optimal solution for the network provider. In addition, both network slicing resource allocation solutions proposed in~\cite{Jiang2016Network} and~\cite{Soliman2016QoS} are heuristic methods with only radio resource taken into consideration. Thus, these solutions may not be appropriate to implement in dynamic network slicing resource allocation systems with a wide range of resource demands and services.

To deal with the dynamic of services, e.g., users' resource demands and their occupation time, the authors in~\cite{Sciancalepore} proposed a model to predict the future demand of slices, thereby maximizing the system resource utilization for the provider. The key idea of this approach is to use the Holt-Winters approach~\cite{Holt_Winters_theory} to predict network slices' demands through tracking the traffic usage of users in the past. However, the accuracy of this prediction depends largely on the heavy-tailed distribution functions along with many control parameters such as scale factor, least-action trip planning, and potential gain. Furthermore, this approach only considers the short-term reward for the provider, and thus the long-term profit may not be able to obtain. Therefore, the authors in~\cite{Bega2017Optimising},~\cite{Aijaz2017Hap}, and~\cite{ToshibaPatent} proposed reinforcement learning algorithms to address these problems. Among dynamic resource allocation methods, reinforcement learning has been considering to be the most effective way to maximize the long-term reward for dynamic systems as this method allows the network controller to adjust its actions in a real-time manner to obtain the optimal policy through the trial-and-error learning process~\cite{Sutton1998Reinforcement}. However, this method often takes a long period to converge to the optimal solution, especially for a large-scale system.

In all aforementioned work, the authors considered optimizing only radio resources, while other resources are completely ignored. However, as stated in~\cite{whitepaper,introduce,FoukasNetworkslicing2017}, a typical network slice is composed of three main components, i.e., radio, computing, and storage. Consequently, considering only radio resources when orchestrating slices may be not able to achieve the optimal solution. Specifically, in Fig.~\ref{Fig.compare}, we show an illustration to demonstrate the inefficiency of optimizing only radio resources. Therefore, in this paper, we introduce a Semi-Markov decision processes (SMDP) framework~\cite{Puterman_1994_Book} which allows the network provider effectively allocate all three resources, i.e., radio, computing, and storage, to the users in a real-time manner. However, when we jointly consider combinatorial resources, i.e., radio, storage, and computing resources, together with the uncertainty of demands, the optimization problem becomes very complex as we need to simultaneously deal with a very large state space with multi-dimension and real-time dynamic decisions. Thus, we propose a novel network slicing framework with an advanced deep learning architecture using two streams of fully connected hidden layers, i.e., deep dueling neural network~\cite{Wang2015Dueling}, combined with Q-learning method to effectively address this problem. Our preliminary simulation results~\cite{ICC} show that our proposed approach not only can effectively deal with the dynamic of the system but also significantly improve the system performance compared with all other current network slicing resource allocation approaches. It is worth noting that the VNF placement, routing, and connectivity resource allocation problems have been well investigated in the literature. Hence, in this paper, we focus on dealing with the uncertainty, dynamics, and heterogeneity of slice requests. Note that, our system model can be straightforwardly extended to the case with diverse connectivity among servers and data centers by accommodating additional states to the system state space. Note that, our proposed framework can handle very well a large state space (one of our key contributions in this paper).

%======================================================
%======================================================
\subsection{Main Contributions}

The main contributions of this paper are as follows:

\begin{itemize}
	\item We develop a dynamic network resource management model based on semi-Markov decision process framework which allows the network provider to jointly allocate computing, storage, and radio resources to different slice requests in a real-time manner and maximize the long-term reward under a number of available resources.
	
	\item To find the optimal policy under the uncertainty of slice service demands, we deploy the Q-learning algorithm which can achieve the optimal solution through reinforcement learning processes. However, the Q-learning algorithm may not be able to effectively achieve the optimal policy due to the curse-of-dimensionality problem when we jointly optimize multiple resources concurrently. Thus, we develop a deep double Q-learning approach which utilizes the advantage of a neural network to train the learning process of the Q-learning algorithm, thereby attaining much better performance than that of the Q-learning algorithm. 

	\item To further enhance the performance of the system, we propose the novel network slicing approach with the deep dueling neural network architecture~\cite{Wang2015Dueling}, which can outperform all other current reinforcement learning techniques in managing network slicing. The key idea of the deep dueling is using two streams of fully connected hidden layers to concurrently train the learning process of the Q-learning algorithm, thereby improving the training process and achieving an outstanding performance for the system.

	\item Finally, we perform extensive simulations with the aim of not only demonstrating the efficiency of proposed solutions in comparison with other conventional methods but also providing insightful analytical results for the implementation of the system. Importantly, through simulation results, we demonstrate that our proposed framework can improve the performance of the system up to 40\% compared with other current approaches.
\end{itemize} 

The rest of the paper is organized as follows. Section~\ref{sec:sysmodel} and Section~\ref{sec:problem_formulation} describe the system model and the problem formulation, respectively. Section~\ref{sec:Qlearning} introduces the Q-learning algorithm. Further on, we present the deep double Q-learning and deep dueling algorithms in Section~\ref{sec:DeepQlearning}. Evaluation results are then discussed in Section~\ref{sec:evaluation}. Finally, conclusions and future works are given in Section~\ref{sec:conclusion}.

%====================================================================================
%====================================================================================
\section{System Model}
\label{sec:sysmodel}

\begin{figure*}[!]
	\centering
	\includegraphics[scale=0.15]{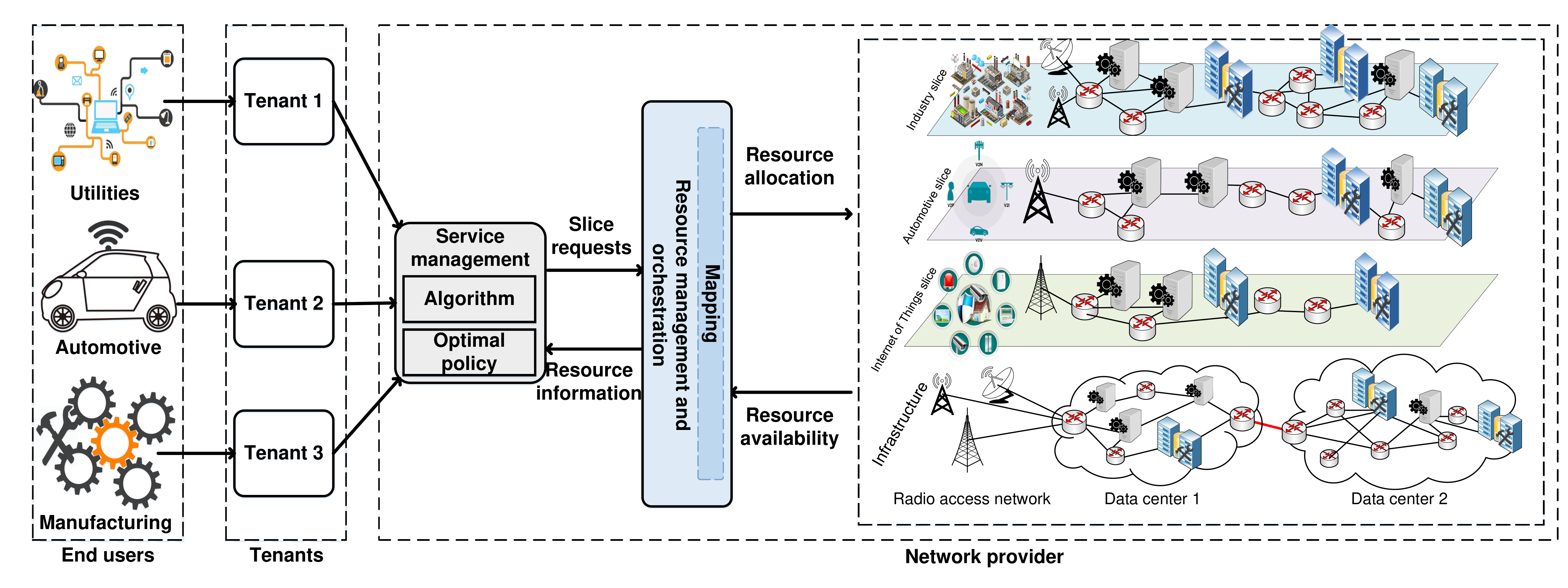}
	\caption{Network resource slicing system model.}
	\label{fig:system_model}
\end{figure*}

In Fig.~\ref{fig:system_model}, we consider a general network slicing model with three major parties~\cite{Sciancalepore},~\cite{Bega2017Optimising},~\cite{FoukasNetworkslicing2017}: 
\begin{itemize}
	\item \emph{Network provider:} is the owner of the network infrastructure who provides resource slices including radio, computing, and storage, to the tenants. 
	\item \emph{Tenants:} request and lease resource slices to meet service demands of their subscribers.
	\item \emph{End users:} run their applications on the slices of the above subscribed tenants.
\end{itemize}

We consider three tenants corresponding to three popular classes of services, i.e., utilities, automotive, and manufacturing, as shown in Fig.~\ref{fig:system_model}. Each class of service possesses some specific features regarding its functional, behavioral perspective, and requirements. For example, a vehicle may need an ultra-reliable slice for telemetry assisted driving~\cite{introduce}. For slices requested from industry, security, resilience, and reliability of services are of higher priority~\cite{Bihannic2018Network},~\cite{5GHuawei}. Thus, when a tenant sends a network slice request to the network provider, the tenant will specify resources requested and additional service requirements, e.g., security and reliability (defined in the slice blueprint). As a result, tenants may pay different prices for their requests, depending on their service demands. Upon receiving a slice request, the service management component (in Fig.~\ref{fig:system_model}) analyzes the requirements and makes a decision to accept or reject the request based on its optimal policy.

The service management block consists of two components: (i) the optimal policy and (ii) the algorithm. When a slice request arrives at the system, the optimal policy component will make a decision, i.e., accept or reject, based on the (current) optimal policy obtained by the algorithm component. As the decision can be made immediately, the decision latency is virtually zero. For the algorithm component, the optimal policy is calculated and updated periodically. It is worth noting that our algorithm observes the results after performing the decision, and uses the observations together with the characteristics of slice requests for its training process. By doing so, our algorithm can learn from previous experience and is able to deal with the uncertainty of slice requests. If the request is accepted, the service management will transfer the slice request to the resource management and orchestration (RMO) block to allocate resources. Once a slice request is accepted, the network provider will receive an immediate reward (the amount of money) paid by the tenant for granted resources and services.

In practice, the network provider may possess multiple data centers for network slicing services. Each data center contains a set of servers with diverse resources, e.g., computing and storage, which are used to support VNFs services. Servers in the data center are connected together, and the data centers are connected via backhaul links. Then, the network slicing and resource allocation processes to each slice are taken place as follows.

\begin{itemize}
	\item A slice request is associated with the \textit{network slice blueprint} (i.e., a template) that describes the structure, configuration, and workflow for instantiating and controlling the network slice instance for the service during its life cycle~\cite{FoukasNetworkslicing2017},~\cite{NFV1}-\cite{NFV3}. The service/slice instance includes a set of network functions and resources to meet the end-to-end service requirements.
	\item When a slice request arrives at the system, the orchestrator will interpret the blueprint~\cite{NFV1}-\cite{NFV3}. In particular, all information of the infrastructure such as (i) NFV services provided by servers, (ii) resources availability at servers, and (iii) the connectivities among servers are checked.
	\item Based on the aforementioned information, the orchestrator will find the optimal servers and links to place VNFs to meet the required end-to-end services of the slice (i.e., VNF placement procedure).
	\item During the life cycle of the slice, the orchestrator can change the allocated computing and storage resources by using the scaling-in and scaling-out mechanisms. In addition, the connectivity among VNFs and their locations can be changed when there is no sufficient resources or the behavior of the slice is changed, e.g., update, migrate, or terminate the network slice.
\end{itemize}

Note that the VNF placement, routing, and connectivity resource allocation problems have been well investigated in the literature, e.g.,~\cite{Fischer2013Virtual}-\cite{Awerbuch1993Throughput}. For example, in~\cite{Awerbuch1993Throughput}, the AAP algorithm is introduced to admit and route connection requests by finding possible paths satisfying cost criteria. Instead of focusing on VNF placement, routing, and connectivity resource allocation problems, in this paper, we mainly focus on dealing with the uncertainty, dynamics, and heterogeneity of slice requests. Thus, we consider a simplified yet practical model and propose the novel framework using the deep dueling neural network architecture~\cite{Wang2015Dueling} to address the aforementioned problems, which are the main aims and key contributions of this work. Specifically, we assume that there are $C$ classes of slices, denoted by $\mathcal{C}= \{1,\ldots, c,\ldots, C\}$. Each slice from class $c$ requires $r^{re}_c$, $\omega^{re}_c$, and $\delta^{re}_c$ units of radio, computing, and storage resources, respectively. If a slice request from class $c$ is accepted, the provider will receive an immediate reward $r_c$. The maximum radio, computing, and storage resources of the network provider are denoted by $\Theta$, $\Omega$, and $\Delta$ units, respectively. Let $n_c$ denote the number of slices from class $c$ being simultaneously run/served in the system. At any time, the following resource constraints guarantee that the allocated resources do not exceed the available resources of the infrastructure:
\begin{equation}
\label{eq:contraints}
\Theta \geq \sum_{c=1}^{C} r_c^{re}n_c, \quad \Omega \geq \sum_{c=1}^{C} \omega_c^{re}n_c, \text{ and } \Delta \geq \sum_{c=1}^{C} \delta_c^{re}n_c.
\end{equation}

%====================================================================================
%====================================================================================
\section{Problem Formulation}
\label{sec:problem_formulation}

To maximize the long-term return for the provider while accounting for the real-time arrivals of slice requests, we recruit the semi-Markov decision process (SMDP)~\cite{Puterman_1994_Book}. An SMDP is defined by a tuple $<t_i, \mathcal{S}, \mathcal{A}, \mathcal{L}, r>$ where $t_i$ is an decision epoch, $\mathcal{S}$ is the system's state space, $\mathcal{A}$ is the action space, $\mathcal{L}$ captures the state transition probabilities and the state sojourn time, and $r$ is the reward function. Unlike discrete Markov decision processes where decisions are made in every time slots, in an SMDP, we only need to make decisions when an event occurs. This makes the SMDP framework more effective to capture real-time network slicing systems.%Here, $i=\{1,2,\ldots\}$ is the indicator to illustrate the decision epoch. In the following, we will discuss more details on parameters used in the formulation. 

%======================================================
%======================================================
\subsection{Decision Epoch}

%A decision epoch is defined as the interval between two successive decisions of the system. 
Under our network slicing system model, the provider needs to make a decision upon receiving requests from tenants. Thus, the decision epoch can be defined as the inter-arrival time between two successive slice requests.

%======================================================
%======================================================
\subsection{State Space}

The system state $\mathbf{s}$ of the SMDP at the current decision epoch captures the number of slices $n_c$ from a given class $c$ ($\forall c\in\mathcal{C}$) being simultaneously run/served in the system. Formally, we define $\mathbf{s}$ as an $1\times C$ vector:
\begin{equation}
\mathbf{s} \triangleq [n_1,\ldots,n_c,\ldots n_C].
\end{equation}
Given the network provider's resource constraints in \eqref{eq:contraints}, the state space $\mathcal{S}$ of all possible states $\mathbf{s}$ is defined as:
\begin{equation}\label{eq:state_space}
\begin{split}
\mathcal{S} \triangleq 	\Bigg\{  \mathbf{s} = [n_1,\ldots,n_c,\ldots n_C] : \Theta \geq \sum_{c=1}^{C} r_c^{re}n_c; \\
\Omega \geq \sum_{c=1}^{C} \omega_c^{re}n_c; \Delta \geq \sum_{c=1}^{C} \delta_c^{re}n_c\Bigg\}.
\end{split}
\end{equation} 

At the current system state $\mathbf{s}$, we define the \emph{event vector} $\mathbf{e} \triangleq [e_1,\ldots, e_c, \ldots, e_C]$ with $e_c\in\{1,-1,0\},\forall c\in \mathcal{C}$. $e_c$ equals to ``$1$'' if a new slice request from class $c$ arrives, $e_c$ equals to ``$-1$'' if a slice's resources are being released (also referred to as a slice completion/departing) to the system's resource, and $e_c$ equals ``$0$'' otherwise (i.e., no slice request arrives nor completes/departs from the system). The set $\mathcal{E}$ of all the possible events is then defined as follows:
\begin{equation}
\mathcal{E} \triangleq \Big\{ \mathbf{e}:e_c \in \{-1, 0, 1\}; \sum_{c=1}^{C}|e_c| \leq 1\Big\},
\end{equation}
where the \emph{trivial event} $\mathbf{e^*}\triangleq (0, \ldots, 0)\in \mathcal{E}$ means no request arrival or completion/departing from all $C$ classes.

%The system state $\mathbf{s}$ at the current decision epoch can now be defined as  $\mathbf{s}\triangleq (\mathbf{n}, \mathbf{e})$. Accordingly, the state space $\mathcal{S}$ of the system is defined as follows:
%\begin{equation}
%\label{eq:state_space}
%\mathcal{S} \triangleq 	\Big\{ \mathbf{s}= (\mathbf{n}, \mathbf{e}): \mathbf{n} \in \mathcal{N};\mathbf{e} \in \mathcal{E} \Big\}.
%\end{equation}

%======================================================
%======================================================
\subsection{Action Space} 
At state $\mathbf{s}$, if a slice request arrives (i.e., there exists $c\in \mathcal{C}$ such that $e_c=1$), the network provider can choose either to accept or reject this request to maximize its long-term return. Let $a_{\mathbf{s}}$ denote the action to be taken at state $\mathbf{s}$ where $a_{\mathbf{s}} = 1$ if an arrival slice is accepted and $a_{\mathbf{s}} = 0$ otherwise. The state-dependent action space $\mathcal{A}_{\mathbf{s}}$ can be defined by:
\begin{equation}
\mathcal{A}_{\mathbf{s}} \triangleq \{ a_{\mathbf{s}}\}= \{0, 1\}.
\end{equation}

%======================================================
%======================================================
\subsection{State Transition Probability}

As aforementioned, in this work, we propose reinforcement learning approaches which can obtain the optimal policy for the network provider without requiring information from the environment (to cope with the uncertain demands and dynamics of slice requests). However, to lay a theoretical foundation and to evaluate the performance of our proposed solutions, we first assume that the arrival process of slice requests from class $c$ follows the Poisson distribution with mean rate $\lambda_c$ and its network resource occupation time follows the exponential distribution with mean $1/\mu_c$. The assumptions allow us to analyze the dynamics of the SMDP, which is characterized by the state transition probabilities of the underlying Markov chain. In particular, our SMDP model consists of a renewal process and a continuous-time Markov chain $\{X(t: t \geq 0)\}$ in which the sojourn time in a state is a continuous random variable. We then can adopt the uniformization technique~\cite{Gallager1995Discrete} to determine the probabilities for events and derive the transition probabilities $\mathcal{L}$. As shown in Fig.~\ref{fig:uniformization}, the uniformization technique transforms the original continuous-time Markov chain $\{X(t): t \geq 0\}$ into an equivalent stochastic process $\{\overline{X}(t), t \geq 0\}$ in which the transition epochs are generated by a Poisson process $\{N(t): t \geq 0\}$ at a uniform rate and the state transitions are governed by the discrete-time Markov chain $\{\overline{X}_n\}$~\cite{Tijms2003Stochastic},~\cite{KallenbergMDP}. The details of the uniformization technique are as the following.
\captionsetup[figure]{font={color=black}}
\begin{figure}[!]
	\centering
	\includegraphics[scale=0.25]{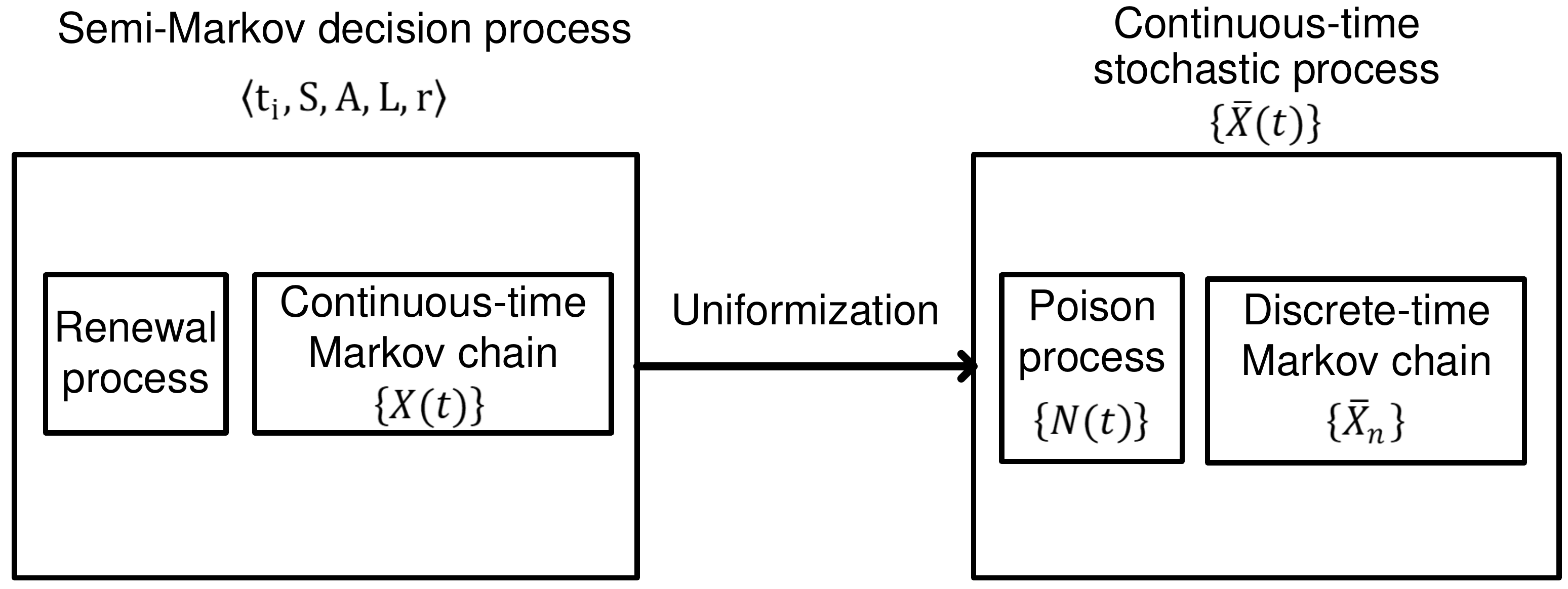}
	\caption{Uniformization technique.}
	\label{fig:uniformization}
\end{figure}

Our Markov chain $\{X(t)\}$ can be considered as a time-homogeneous Markov chain. Suppose that $\{X(t)\}$ is in state $\mathbf{s}$ at the current time $t$. If the system leaves state $\mathbf{s}$, it transfers to state $\mathbf{s'}$ ($\neq \mathbf{s}$) with probability $p_{\mathbf{s}, \mathbf{s'}}(t)$. The probability that the process will leave state $\mathbf{s}$ in the next $\Delta t$ to state $\mathbf{s'}$ is expressed as follows:
\begin{equation}
\begin{aligned}
P\{X(t + \Delta t) & =\mathbf{s'}| X(t) =\mathbf{s}\}   \\
& = \left\{	\begin{array}{ll}
z_{\mathbf{s}}\Delta t \times p_{\mathbf{s}, \mathbf{s'}}(t) + o(\Delta t),	& \mathbf{s'} \neq \mathbf{s},	\\
1- z_{\mathbf{s}}\Delta t + o(\Delta t),	&	 \mathbf{s'} = \mathbf{s},
\end{array}	\right.
\end{aligned}
\end{equation}
as $\Delta t \rightarrow 0$ and $z_{\mathbf{s}}$ is the occurrence rate of the next event expressed as follows:
\begin{equation}
z_{\mathbf{s}} = \sum_{c=1}^{C}(\lambda_c+n_c\mu_c).
\end{equation}
In the uniformization technique, we consider that the occurrence rate $z_{\mathbf{s}}$ of the states are identical, i.e., $z_{\mathbf{s}}= z$ for all $\mathbf{s}$. Thus, the transition epochs can be generated by a Poisson process with rate $z$. To formulate the uniformization technique, we choose a number $z$ with
\begin{equation}
z= \max_{\mathbf{s} \in \mathcal{S}} z_{\mathbf{s}}.
\end{equation}

Now, we define a discrete-time Markov chain $\{\overline{X}_n\}$ whose one-step transition probabilities $\overline{p}_{\mathbf{s},\mathbf{s'}}(t)$ are given by:
\begin{equation}
\label{eq:overlinep}
\overline{p}_{\mathbf{s},\mathbf{s'}}(t) 	=	\left\{	\begin{array}{ll}
(z_{\mathbf{s}}/z)p_{\mathbf{s},\mathbf{s'}}(t),	& \mathbf{s'} \neq \mathbf{s},	\\
1-z_{\mathbf{s}}/z,	&	\mbox{otherwise},
\end{array}	\right.
\end{equation}
for all $\mathbf{s} \in \mathcal{S}$. Let $\{N(t), t \geq 0\} $ be a Poisson process with rate $z$ such that the process is independent of the discrete-time Markov chain $\{\overline{X}_n\}$. We then define the continuous-time stochastic process $\{\overline{X}(t), t \geq 0\}$ as follows:
\begin{equation}
\label{convert}
\overline{X}(t) = \overline{X}_{N(t)}, t \geq 0.
\end{equation}
Equation~(\ref{convert}) represents that the process $\{\overline{X}(t)\}$ makes state transitions at epochs generated by a Poisson process with rate $z$ and the state transitions are governed by the discrete-time Markov chain $\{\overline{X}_n\}$ with one-step transition probabilities $\overline{p}_{\mathbf{s},\mathbf{s'}}(t)$ in~(\ref{eq:overlinep}). When the Markov chain $\{\overline{X}_n\}$ is in state $\mathbf{s}$, the system leaves to next state with probability $z_{\mathbf{s}}/z$ and is a self-transition with probability $1-z_{\mathbf{s}}/z$. In fact, the transitions out of state $\mathbf{s}$ are delayed by a time factor of $z/z_{\mathbf{s}}$, while a factor of $z_{\mathbf{s}}/z$ corresponds to the time until a state transition from state $\mathbf{s}$. In addition, in our system model there is no terminal state, i.e., the discrete-time Markov chain describing the state transitions in the transformed process has to allow for self-transitions leaving the state of the process unchanged. Therefore, the continuous $\{\overline{X}(t)\}$ is probabilistically identical to the original continuous-time Markov chain $\{X(t)\}$. This statement can be expressed as the following equation:
\begin{equation}
\begin{aligned}
& P\{\overline{X}(t+\Delta t)=\mathbf{s'}|\overline{X}(t)=\mathbf{s}\} = z \Delta t \times \overline{p}_{\mathbf{s},\mathbf{s'}} + o(\Delta t)\\
&=z_{\mathbf{s}} \Delta t \times p_{\mathbf{s},\mathbf{s'}} + o(\Delta t)\\
&=q_{\mathbf{s},\mathbf{s'}}\Delta t + o(\Delta t)\\
&=P\{X(t+\Delta t) = \mathbf{s'}| X(t) = \mathbf{s}\} \mbox{ for } \Delta t \rightarrow 0, \forall \mathbf{s},\mathbf{s'} \in \mathcal{S} \\
& \mbox{ and } \mathbf{s} \neq \mathbf{s'},
\end{aligned}
\end{equation}
where $q_{\mathbf{s},\mathbf{s'}}$ is the infinitesimal transition rate of the continuous-time Markov chain $\{X(t)\}$ and is expressed as follows:
\begin{equation}
q_{\mathbf{s},\mathbf{s'}} = z_{\mathbf{s}}p_{\mathbf{s},\mathbf{s'}}, \forall \mathbf{s}, \mathbf{s'} \in \mathcal{S} \quad \mbox{and} \quad \mathbf{s'} \neq \mathbf{s}.
\end{equation}
Clearly, in our system, the occurrence rate of the next event $z_{\mathbf{s}} = \sum_{c=1}^{C}(\lambda_c+n_c\mu_c)$ are positive and bounded in $\mathbf{s} \in \mathcal{S}$. Thus, it is proved that the infinitesimal transition rates determine a unique continuous-time Markov chain $\{X(t)\}$~\cite{Tijms2003Stochastic}. We then make a necessary corollary as follows:
\begin{cor}
	\label{cor}
	The probabilities $p_{\mathbf{s},\mathbf{s'}}(t)$ are given by:
	\begin{equation}
	\label{complex1}
	p_{\mathbf{s},\mathbf{s'}}(t) = \sum_{n=0}^{\infty}e^{-zt}\frac{{zt}^n}{n!} \overline{p}_{\mathbf{s},\mathbf{s'}}^{(n)}, \forall \mathbf{s}, \mathbf{s'} \in \mathcal{S} \mbox{ and } t \ge 0,
	\end{equation}
	where the probabilities $\overline{p}_{\mathbf{s},\mathbf{s'}}^{(n)}$ can be recursively computed from
	\begin{equation}
	\label{complex2}
	\overline{p}_{\mathbf{s},\mathbf{s'}}^{(n)} = \sum_{k \in \mathcal{S}}^{}\overline{p}_{\mathbf{s},k}^{(n-1)}\overline{p}_{k,\mathbf{s'}}, n = 1,2,\ldots
	\end{equation}
	starting with $\overline{p}_{\mathbf{s},\mathbf{s}}^{(0)}=1$ and $\overline{p}_{\mathbf{s}, \mathbf{s'}}^{0}=0$ $\forall \mathbf{s'} \neq \mathbf{s}$ .
\end{cor}

In the next theorem, we prove that two processes $\{\overline{X}(t)\}$ and $\{X(t)\}$ are probabilistically equivalent.
\begin{theorem}
	\label{theorem_equivalent}
	 $\{\overline{X}(t)\}$ and $\{X(t)\}$ are probabilistically equivalent as
	\begin{equation}
	p_{\mathbf{s},\mathbf{s'}}(t) = P\{\overline{X}(t)=\mathbf{s'}|X(0)=\mathbf{s}\}, \forall \mathbf{s},\mathbf{s'} \in \mathcal{S} \mbox{ and } t \le 0.
	\end{equation}
\end{theorem}
The proof of Theorem~\ref{theorem_equivalent} is given in Appendix~\ref{appendix:pro_equivalent}.$\hfill\blacksquare$

From~(\ref{complex1}) and~(\ref{complex2}), the computational complexity of uniformization method is derived as $O(vt|\mathcal{S}|^2)$, where $|\mathcal{S}|$ is the number of states of the system. Based on $z$ and $z_{\mathbf{s}}$, we can determine the probabilities for events as follows. The probability for an arrival slice from class $c$ occurring in the next event $\mathbf{e}$ equals $\lambda_c/z$. The probability for a departure slice from class $c$ occurring in the next event $\mathbf{e}$ equals $n_c\mu_c/z$, and the probability for a trivial event occurring in the next event $\mathbf{e}$ is $1-z_{\mathbf{s}}/z$. Hence, we can derive the transition probability $\mathcal{L}$.

%======================================================
%======================================================
\subsection{Reward Function}

The immediate reward after action $a_{\mathbf{s}}$ is executed at state $\mathbf{s} \in \mathcal{S}$ is defined as follows:
\begin{equation}
r(\mathbf{s},a_{\mathbf{s}})	=	\left\{	\begin{array}{ll}
r_c,	&	\mbox{if $e_c = 1$, $a_{\mathbf{s}}=1$, and $\mathbf{s'} \in \mathcal{S}$},	\\
0,	&	\mbox{otherwise}.
\end{array}	\right.
\end{equation}
At state $\mathbf{s}$, if an arrival slice is accepted, i.e., $a_{\mathbf{s}}=1$, the system will move to next state $\mathbf{s'}$ and the network provider receives an immediate reward $r_c$. In contrast, the immediate reward is equal to $0$ if an arrival slice is rejected or there is no slice request arriving at the system. The value of $r_c$ represents the amount of money paid by the tenant based on resources and additional services required.

As our system's statistical properties are time-invariant, i.e., stationary, the decision policy $\pi$ of the SMDP model, which is a pure strategy, i.e., accept or reject an arrival request, can be defined as a time-invariant mapping from the state space to the action space: $\mathcal{S} \rightarrow \mathcal{A}_{\mathbf{s}}$. Thus, the long-term average reward starting from a state $\mathbf{s}$ can be formulated as follows:
\begin{equation}
\label{eq:average_reward}
\mathcal{R}_{\pi}(\mathbf{s}) = \lim\limits_{K \rightarrow \infty} \frac{\mathbb{E}\{\sum_{k=0}^{K} r({\mathbf{s}}_k, \pi(\mathbf{s}_k)) | \mathbf{s}_0 = \mathbf{s}\}}{\mathbb{E}\{\sum_{k=0}^{K}\tau_k | \mathbf{s}_0 = \mathbf{s}\}}, \forall \mathbf{s} \in \mathcal{S},
\end{equation}
where $\tau_k$ is the time interval between the $k$-th and $(k+1)$-th decision epoch, $r$ is the immediate reward of the system, and $\pi(\mathbf{s})$ is the action corresponding to the policy $\pi$ at state $\mathbf{s}$. 

In the following theorem, we will prove that the limit in Equation~(\ref{eq:average_reward}) exists. 
\begin{theorem}
	\label{theo:reward}
	Given the state space $\mathcal{S}$ is countable and there is a finite number of decision epochs within a certain considered finite time, we have:
	\begin{equation}
	\begin{aligned}
	\label{eq:reward_theorem}
	\mathcal{R}_{\pi}(\mathbf{s}) &= \lim\limits_{K \rightarrow \infty} \frac{\mathbb{E}\{\sum_{k=0}^{K} r({\mathbf{s}}_k, \pi(\mathbf{s}_k)) | \mathbf{s}_0 = \mathbf{s}\}}{\mathbb{E}\{\sum_{k=0}^{K}\tau_k | \mathbf{s}_0 = \mathbf{s}\}}\\
	&= \frac{\overline{\mathcal{L}}_\pi r(\mathbf{s},\pi(\mathbf{s}))}{\overline{\mathcal{L}}_\pi y(\mathbf{s},\pi(\mathbf{s}))}, \forall \mathbf{s} \in \mathcal{S},
	\end{aligned}
	\end{equation}
	where $y(\mathbf{s},\pi(\mathbf{s}))$ is the expected time interval between adjacent decision epochs when action $\pi(\mathbf{s})$ is taken under state $\mathbf{s}$, and
	\begin{equation}
	\overline{\mathcal{L}}_\pi = \lim\limits_{K \rightarrow \infty} \frac{1}{K}\sum_{k=0}^{K-1} \mathcal{L}_\pi^k,
	\end{equation}
	where $\mathcal{L}_\pi^k$ and $\overline{\mathcal{L}}_\pi$ are the transition probability matrix and the limiting matrix of the embedded Markov chain for policy $\pi$, respectively.
\end{theorem}
\begin{proof}
	We first have the following lemma.
	\begin{lemma}
		\label{Lem:exist}
		Given the the transition probability matrix $\mathcal{L}_\pi$, the limiting matrix $\overline{\mathcal{L}}_\pi$ exists.
	\end{lemma}
	The proof of Lemma~\ref{Lem:exist} is given in Appendix~\ref{app:reward}. 
	
	Since $\mathcal{L}_\pi$ is the transition probability matrix, $\overline{\mathcal{L}}_\pi$ exists as stated in Lemma~\ref{Lem:exist}. As the total of probabilities that the system transform from state $\mathbf{s}$ to other states is equal to $1$, we have:
	\begin{equation}
	\label{eq:sum1}
	\sum_{\mathbf{s'}\in \mathcal{S}}^{}\overline{\mathcal{L}}_\pi(\mathbf{s'}|\mathbf{s}) = 1.
	\end{equation}
	From~(\ref{eq:sum1}), we derive $\mathcal{L}_\pi^n$ and $\overline{\mathcal{L}}_\pi$ as follows:
	\begin{equation}
	\begin{aligned}
	\overline{\mathcal{L}}_\pi r(\mathbf{s}, \pi(\mathbf{s})) = \lim\limits_{K \rightarrow \infty} \frac{1}{K+1}\mathbb{E} \{\sum_{k=0}^{K}r(\mathbf{s}_k, \pi(\mathbf{s}_k))\}, \forall \mathbf{s} \in \mathcal{S}, \\
%	\end{equation}
%	and
%	\begin{equation}
	\overline{\mathcal{L}}_\pi y(\mathbf{s}, \pi(\mathbf{s})) = \lim\limits_{N \rightarrow \infty} \frac{1}{N+1}\mathbb{E}\{\sum_{n=0}^{N}\tau_n\}, \forall \mathbf{s} \in \mathcal{S}.
	\end{aligned}
	\end{equation}
	Therefore, (\ref{eq:reward_theorem}) is obtained by taking ratios of these two quantities. Note that the limit of the ratios equals to the ratio of the limits, and that, when taking the limit of the ratios, the factor $\frac{1}{K+1}$ can be removed from the numerator and denominator.
\end{proof}

Note that in our SMDP model, the embedded Markov chain is unichain including a single recurrent class and a set of transient states for all pure policies $\pi$~\cite{Puterman_1994_Book}. Hence, the average reward $\mathcal{R}_{\pi}(\mathbf{s})$ is independent to the initial state, i.e., $\mathcal{R}_{\pi}(\mathbf{s})=\mathcal{R}_{\pi}, \forall \mathbf{s} \in \mathcal{S}$. The average reward maximization problem is then written as:
\begin{eqnarray}
\label{eq:optimization}
\max_{\pi}	& &	\mathcal{R}_{\pi} = \frac{\overline{\mathcal{L}}_\pi r(\mathbf{s},\pi(\mathbf{s}))}{\overline{\mathcal{L}}_\pi y(\mathbf{s},\pi(\mathbf{s}))}	\\
\mbox{s.t.}			& & \sum_{\mathbf{s'}\in \mathcal{S}}^{}\overline{\mathcal{L}}_\pi(\mathbf{s'}|\mathbf{s}) = 1, \forall \mathbf{s} \in \mathcal{S} \nonumber.
\end{eqnarray}
Our objective is to find the optimal admission policy that maximizes the average reward of the network provider, i.e.,
\begin{equation}
\pi^* = \argmax_\pi \mathcal{R}_{\pi}.
\end{equation}

As aforementioned, the network resources may come from multiple data centers with diverse connectivity among servers and data centers. In such a case, the above formulation can be straightforwardly extended by accommodating additional states to the system state space. Specifically, one can define the system state space that includes (i) services of requests together with their corresponding resources and order, i.e., the network slice blueprint, (ii) available resources and services at the servers, and (iii) connectivity among servers and data centers. This means that we just need to increase the state space, compared with the current state space (with three types of resources, as an example) in the current formulation, to capture additional resources and options. Then, the proposed admission/rejection framework can be implemented at the orchestrator to allocate the available resources to requested slices. Specifically, based on this state space, when a slice request arrives, the orchestrator is able to check whether to allocate an optimal possible link to the request (using existing network slicing mechanisms) and then makes a decision to accept or reject the request. In addition, after making a decision to allocate the resources for a slice request (i.e., after the initial deployment of VNFs), if the running slice requires to add more resources or remove some resources (i.e., scaling out or scaling in, respectively), we can consider some new events (i.e., requests to add or remove resources from running slices) to the system state space. Again, this implies that we only need to add more states to the system state space of the current model. Note that, the action space will be kept the same, i.e., only two actions (accept or reject), and we just need to set new rewards for accepting/rejecting requests from running slices in the problem formulation.
%After that, all our proposed techniques, i.e., Q-learning, Deep Q-learning, Deep Double Q-learning, and Deep Dueling, can be applied. Note that, our proposed deep dueling algorithm can handle very well a large state space (one of our key contributions in this paper). The convergence speed obtained by the deep dueling approach is few thousands times faster than that of the conventional Q-learning algorithm.

Note that the problem~(\ref{eq:optimization}) requires environment information, i.e., arrival and completion rates of slice requests, to construct the transition probability matrix $\mathcal{L}$. Nevertheless, due to the uncertain demands and the dynamics of slice requests from tenants, these environment parameters may not be available and can be time-varying. To cope with the demand uncertainty, we consider deep double Q-learning and deep dueling algorithms to find the optimal admission policy at the RMO to maximize the long-term average reward.

%====================================================================================
%====================================================================================
\section{Q-Learning for Dynamic Resource Allocation Under Uncertainty of Slice Service Demands}
\label{sec:Qlearning}

%%======================================================
%%======================================================
%\subsection{Q-Learning Algorithm}

%To deal with the uncertainty of slice service demands, e.g., service time and resource demands, we propose Q-learning algorithm~\cite{Watkins1992QLearning} to help the network provider make optimal decisions, i.e., accept or reject, when a slice request arrives at the system. 
Q-learning \cite{Watkins1992QLearning} is a reinforcement learning technique (to learn environment parameters while they are not available) to find the optimal policy. In particular, the algorithm implements a Q-table to store the value for each pair of state and action. Given the current state, the network provider will make an action based on its current policy. After that, the Q-learning algorithm observes the results, i.e., reward and next state, and updates the value of the Q-table accordingly. In this way, the Q-learning algorithm will be able to learn from its decisions and adjust its policy to converge to the optimal policy after a finite number of iterations~\cite{Watkins1992QLearning}. In this paper, we aim to find the optimal policy $\pi^*:\mathcal{S} \rightarrow \mathcal{A}_{\mathbf{s}}$ for the network provider to maximize its long-term average reward. Specifically, let denote $\mathcal{V}^\pi(\mathbf{s}): \mathcal{S} \rightarrow \mathbb{R}$ as the expected value function obtained by policy $\pi$ from each state $\mathbf{s} \in \mathcal{S}$:
\begin{equation}
\begin{aligned}
\mathcal{V}^\pi(\mathbf{s}) &= \mathbb{E}_\pi \Big [ \sum_{t=0}^{\infty} \gamma^t r_t(\mathbf{s}_t, a_{\mathbf{s}_t})|\mathbf{s}_0=\mathbf{s}\Big ] \\
&=\mathbb{E}_\pi\Big [ r_t(\mathbf{s}_t, a_{\mathbf{s}_t}) + \gamma\mathcal{V}^\pi(\mathbf{s}_{t+1})|\mathbf{s}_0=\mathbf{s}\Big ],
\end{aligned}
\end{equation}
where $0\leq \gamma < 1$ is the discount factor which determines the importance of long-term reward~\cite{Watkins1992QLearning}. In particular, if $\gamma$ is close to 0, the RMO will prefer to select actions to maximize its short-term reward. In contrast, if $\gamma$ is close to 1, the RMO will make actions such that its long-term reward is maximized. $r_t(\mathbf{s}_t, a_{\mathbf{s}_t})$ is the immediate reward achieved by taking action $a_{\mathbf{s}_t}$ at state $\mathbf{s}_t$. Given a state $\mathbf{s}$, policy $\pi(\mathbf{s})$ is obtained by taking action $a_{\mathbf{s}}$ whose the value function is highest~\cite{Watkins1992QLearning}. Since we aim to find optimal policy $\pi^*$, an optimal action at each state can be found through the optimal value function expressed by:
\begin{equation}
\mathcal{V}^*(\mathbf{s}) = \max_{a_{\mathbf{s}}} \Big \{ \mathbb{E}_\pi[r_t(\mathbf{s}_t, a_{\mathbf{s}_t})+ \gamma\mathcal{V}^\pi(\mathbf{s}_{t+1})] \Big\}, \forall \mathbf{s} \in \mathcal{S}.
\end{equation}
The optimal Q-functions for state-action pairs are denoted by:
\begin{equation}
\mathcal{Q}^*(\mathbf{s},a_{\mathbf{s}}) \triangleq r_t(\mathbf{s}_t, a_{\mathbf{s}_t}) + \gamma\mathbb{E}_\pi[\mathcal{V}^\pi(\mathbf{s}_{t+1})], \forall \mathbf{s} \in \mathcal{S}.
\end{equation}
Then, the optimal value function can be written as follows:
\begin{equation}
\mathcal{V}^*(\mathbf{s}) = \max_{a_{\mathbf{s}}} \{ \mathcal{Q}^*(\mathbf{s},a_{\mathbf{s}})\}.
\end{equation}

By making samples iteratively, the problem is reduced to determining the optimal value of Q-function, i.e., $\mathcal{Q}^*(\mathbf{s},a_{\mathbf{s}})$, for all state-action pairs. In particular, the Q-function is updated according to the following rule:
\begin{equation}
\label{Eq:qfunction}
\begin{aligned}
\mathcal{Q}_{t}&(\mathbf{s}_t,a_{\mathbf{s}_t}) = \mathcal{Q}_t(\mathbf{s}_t,a_{\mathbf{s}_t}) \\
&+ \alpha_t\Big [ r_t(\mathbf{s}_t, a_{\mathbf{s}_t}) + \gamma\max_{a_{\mathbf{s}_{t+1}}} \mathcal{Q}_t(\mathbf{s}_{t+1}, a_{\mathbf{s}_{t+1}})- \mathcal{Q}_t(\mathbf{s}_t,a_{\mathbf{s}_t})\Big ].
\end{aligned}
\end{equation}
The key idea behind this rule is to find the temporal difference between the predicted Q-value, i.e., $r_t(\mathbf{s}_t, a_{\mathbf{s}_t}) + \gamma\max_{a_{\mathbf{s}_{t+1}}} \mathcal{Q}_t(\mathbf{s}_{t+1}, a_{\mathbf{s}_{t+1}})$ and its current value, i.e., $\mathcal{Q}_t(\mathbf{s}_t,a_{\mathbf{s}_t})$. In~(\ref{Eq:qfunction}), the learning rate $\alpha_t$ determines the impact of new information to the existing value. During the learning process, $\alpha_t$ can be adjusted dynamically, or it can be chosen to be a constant. However, to guarantee the convergence for the Q-learning algorithm, $\alpha_t$ is deterministic, nonnegative, and satisfies the following conditions~\cite{Watkins1992QLearning}:
\begin{equation}
\label{Eq:rules}
\alpha_t \in [0,1), \sum_{t=1}^{\infty}\alpha_t = \infty, \mbox{ and } \sum_{t=1}^{\infty} ( \alpha_t  )^{2} < \infty.
\end{equation}

%The detail of the Q-learning algorithm is provided in Algorithm~\ref{algorithm0}.
%\begin{algorithm}
%	\caption{The Q-learning Algorithm}
%	\label{algorithm0}
%	\begin{algorithmic}[1]
%		\State \textbf{Inputs:} For each state-action pair $(\mathbf{s}, a_{\mathbf{s}})$, initialize the table entry $\mathcal{Q}(\mathbf{s}, a_{\mathbf{s}})$ arbitrarily, e.g., to zero. Observe the current state $\mathbf{s}$, initialize a value for the learning rate $\alpha$ and the discount factor $\gamma$.
%		\For{\textit{i=1 to T}}
%		\State From the current state-action pair $(\mathbf{s}_t, a_{\mathbf{s}_t})$, execute action $a_{\mathbf{s}_t}$ and obtain the immediate reward $r_t$ 
%		\State and a new state $\mathbf{s}_{t+1}$. Select an action $a_{\mathbf{s}_{t+1}}$ based on the state $\mathbf{s}_{t+1}$ and then update the table entry 
%		\State for $\mathcal{Q}(\mathbf{s}_t, a_{\mathbf{s}_t})$ as follows:
%		\begin{equation}
%		\mathcal{Q}_{t+1}(\mathbf{s}_t,a_{\mathbf{s}_t}) = \mathcal{Q}_t(\mathbf{s}_t,a_{\mathbf{s}_t}) + \alpha_t\Big [ r_t(\mathbf{s}_t, a_{\mathbf{s}_t}) + \gamma\max_{a_{\mathbf{s}_{t+1}}} \mathcal{Q}_t(\mathbf{s}_{t+1}, a_{\mathbf{s}_{t+1}})- \mathcal{Q}_t(\mathbf{s}_t,a_{\mathbf{s}_t})\Big ].
%		\end{equation}
%		\State Replace $\mathbf{s}_t \leftarrow \mathbf{s}_{t+1}$.
%		\EndFor
%		\State {\textbf{Outputs:}} $\pi^*(\mathbf{s}) = \arg\max_{a_{\mathbf{s}}} \mathcal{Q}^*(\mathbf{s},a_{\mathbf{s}})$.
%	\end{algorithmic}
%\end{algorithm}

Based on~(\ref{Eq:qfunction}), the RMO can employ the Q-learning to obtain the optimal policy. Specifically, the algorithm first initializes the table entry $\mathcal{Q}(\mathbf{s}, a_{\mathbf{s}})$ arbitrarily, e.g., to zero for each state-action pair $(\mathbf{s}, a_{\mathbf{s}})$. From the current state $\mathbf{s}_t$, the algorithm will choose action $a_{\mathbf{s}_t}$ and observe results after performing this action. In practice, to select action $a_{\mathbf{s}_t}$, $\epsilon$-greedy algorithm~\cite{Sutton1998Reinforcement},~\cite{Geramifard2013Atutorial} is often used. Specifically, this method introduces a parameter $\epsilon$ which suggests for the controller in choosing a random action with probability $\epsilon$ or select an action that maximizes the $\mathcal{Q}(\mathbf{s}_t,a_{\mathbf{s}_t})$ with probability $1-\epsilon$. Doing so, the algorithm can explore the whole state space. Hence, we need to balance between the exploration time, i.e., $\epsilon$, and the exploitation time, i.e., $1-\epsilon$, to speed up the convergence of the Q-learning algorithm. The algorithm then determines the next state and reward after performing the chosen action and update the table entry for $\mathcal{Q}(\mathbf{s}_t, a_{\mathbf{s}_t})$ based on~(\ref{Eq:qfunction}). Once either all $Q$-values converge or the certain number of iterations is reached, the algorithm will be terminated. This algorithm yields the optimal policy indicating an action to be taken at each state such that $\mathcal{Q}^*(\mathbf{s},a_{\mathbf{s}})$ is maximized for all states in the state space, i.e., $\pi^*(\mathbf{s}) = \argmax_{a_{\mathbf{s}}} \mathcal{Q}^*(\mathbf{s},a_{\mathbf{s}})$. Under~(\ref{Eq:rules}), it was proved in~\cite{Watkins1992QLearning} that the Q-learning algorithm will converge to the optimum action-values with probability one.

Several studies in the literature reported the application of the Q-learning algorithm to address the network slicing problem, e.g., \cite{Bega2017Optimising}. Note that the Q-learning algorithm can efficiently obtain the optimal policy when the state space and action space are small. However, when the state or action space is large, the Q-learning algorithm often converges prohibitively slow. For the combinatorial resources slicing problem in this work, the state space is can be as large as tens of thousands. That makes the Q-learning algorithm practically inapplicable (especially for real-time resource slicing)~\cite{LillicrapContinuous}. In the sequel, leveraging the deep double Q-learning and deep dueling networks, we develop optimal and fast algorithms to overcome this shortcoming.

%====================================================================================
%====================================================================================	
\section{Fast and Optimal Resources Slicing with Deep Neural Network}
\label{sec:DeepQlearning}

%======================================================
%======================================================
\subsection{Deep Double Q-Learning}
In this section, we introduce the deep double Q-learning algorithm to address the slow-convergence problem of Q-learning algorithm. Originally, the deep double Q-learning algorithm was developed by Google DeepMind in 2016~\cite{deep_double} to teach machines to play games without human control. Specifically, the deep double Q-learning algorithm is introduced to further improve the performance of the deep Q-learning algorithm~\cite{Mnih2015Human}. The key idea of the deep double Q-learning algorithm is to select an action by using the primary network. It then uses the target network to compute the target Q-value for the action.%, instead of taking the maximum value of all Q-values as in the deep Q-learning algorithm.

\begin{figure*}[!]
	\centering
	\includegraphics[scale=0.22]{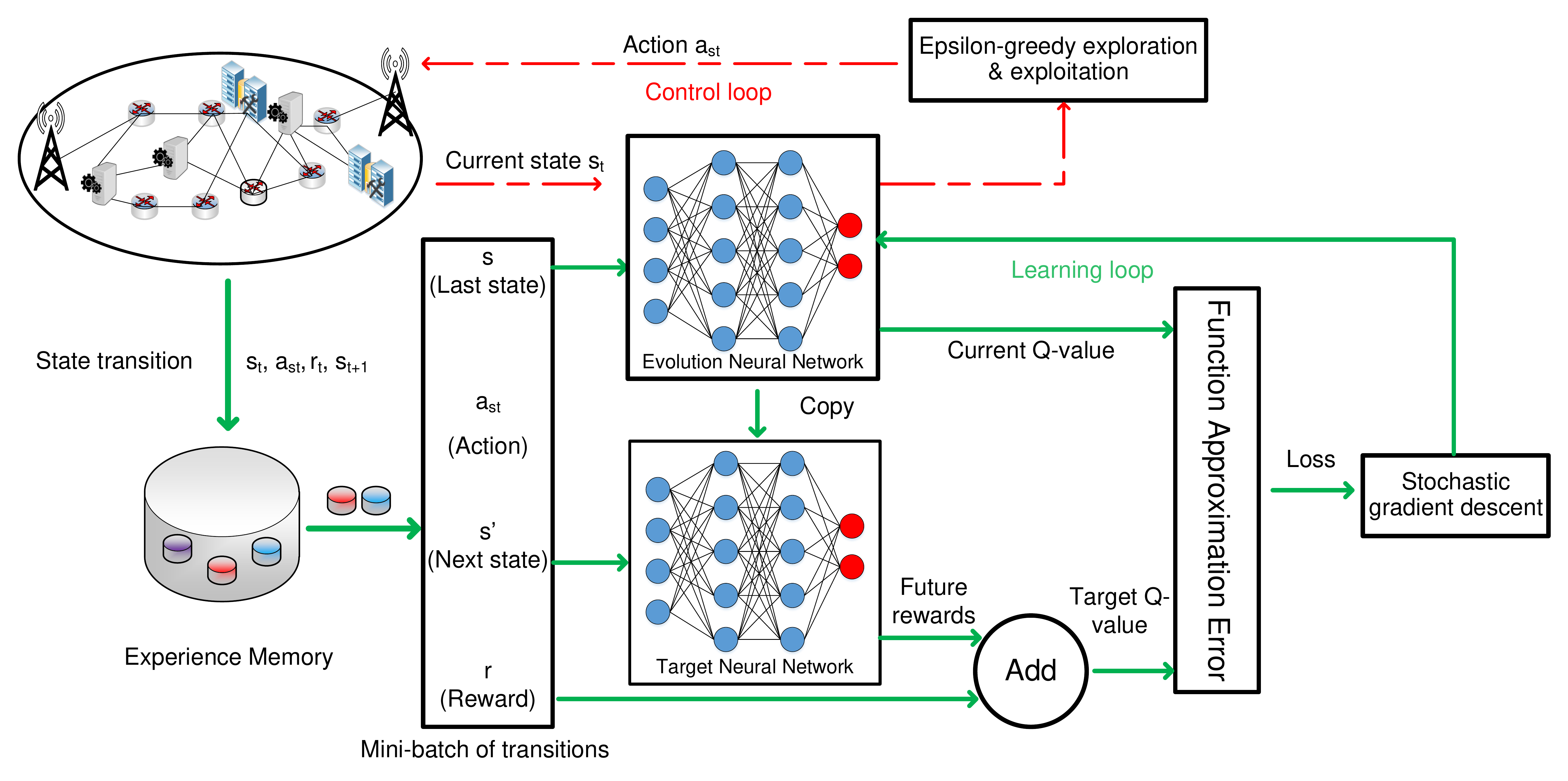}
	\caption{Deep double Q-learning model.}
	\label{Fig.deepqlearning}
\end{figure*}

As pointed in~\cite{Mnih2015Human}, the performance of reinforcement learning approaches might not be stable or even diverges when a nonlinear function approximator is used. This is attributed to the fact that a small change of Q-values may greatly affect the policy. Thereby the data distribution and the correlations between the Q-values and the target values are varied. To address this issue, we use the experience replay mechanism, the target Q-network, and a proper feature set selection:

\begin{itemize}
	\item \textit{Experience replay mechanism:} The algorithm will store transitions $(\mathbf{s}_t, a_{\mathbf{s}_t}, r_t, \mathbf{s}_{t+1})$ in the replay memory, i.e., memory pool, instead of running on state-action pairs as they occur during experience. The learning process is then performed based on random samples from the memory pool. By doing so, the previous experiences are exploited more efficiently as the algorithm can learn them many times. Additionally, by using the experience mechanism, the data is more like independent and identically distributed. That removes the correlations between observations.
	\item \textit{Target Q-network:} In the training process, the Q-value will be shifted. Thus, the value estimations can be out of control if a constantly shifting set of values is used to update the Q-network. This destabilizes the algorithm. To address this issue, we use the target Q-network to frequently (but slowly) update to the primary Q-networks values. That significantly reduces the correlations between the target and estimated Q-values, thereby stabilizing the algorithm.
	\item \textit{Feature set: } For each state, we determine four features including radio, computing, storage, and event trigger, i.e., a slice request arrives. These features are then fed into the deep neural network to approximate the Q-values for each action of a state. As such, all aspects of each state are trained in the deep neural network, resulting in a higher convergence rate.
\end{itemize}
\begin{algorithm}
	\caption{Deep Double Q-learning Based Resources Slicing Algorithm}
	\label{deepqlearning}
	\begin{algorithmic}[1]
		\State Initialize replay memory to capacity $D$.
		\State Initialize the Q-network $\mathcal{Q}$ with random weights $\theta$.
		\State Initialize the target Q-network $\hat{\mathcal{Q}}$ with weight $\theta^-=\theta$.
		\For{\textit{episode=1 to T}}
		\State With probability $\epsilon$ select a random action $a_{\mathbf{s}_t}$, otherwise select $a_{\mathbf{s}_t}=\argmax \mathcal{Q}^*(\mathbf{s}_t, a_{\mathbf{s}_t}; \theta)$
		\State Perform action $a_{\mathbf{s}_t}$ and observe reward $r_t$ and next state $\mathbf{s}_{t+1}$
		\State Store transition $(\mathbf{s}_t, a_{\mathbf{s}_t}, r_t, \mathbf{s}_{t+1})$ in the replay memory
		\State Sample random minibatch of transitions $(\mathbf{s}_j, a_{\mathbf{s}_j}, r_j, \mathbf{s}_{j+1})$ from the replay memory
		%		\If{episode terminates at step j+1}
		%		\State $y_j=r_j$
		%		\Else
		\State $y_j=r_j+\gamma\mathcal{Q}(\mathbf{s}_{j+1},\argmax_{a_{\mathbf{s}_{j+1}}}\hat{\mathcal{Q}}(\mathbf{s}_{j+1}, a_{\mathbf{s}_{j+1}};\theta);\theta^-)$
		%		\EndIf
		\State Perform a gradient descent step on $(y_j-\mathcal{Q}(\mathbf{s}_j, a_{\mathbf{s}_j}; \theta))^2$ with respect to the network parameter $\theta$.
		\State Every $C$ steps reset $\hat{\mathcal{Q}} = \mathcal{Q}$
		\EndFor
	\end{algorithmic}
\end{algorithm}

The details of the deep double Q-learning algorithm is provided in Algorithm~\ref{deepqlearning} and explained more details in the flowchart in Fig.~\ref{Fig.flowchart}. Specifically, as shown in Fig.~\ref{Fig.flowchart}, the training phase is composed of multiple episodes. In each episode, the RMO performs an action and learns from observations corresponding to the taken action. As a result, the RMO needs to tradeoff between the exploration and exploitation processes over the state space. Therefore, in each episode, given the current state, the algorithm will choose an action based on the epsilon-greedy algorithm. The algorithm will start with a fairly randomized policy and later slowly move to a deterministic policy. This means that, at the first episode, $\epsilon$ is set at a large value, e.g., 0.9, and gradually decayed to a small value, e.g., 0.1. After that, the RMO will perform the selected action and observe results, i.e., next state and reward, from taking this action. This transition is then stored in the replay memory for the training process at later episodes.

In the learning process, random samples of transitions from the replay memory will be fed into the neural network. In particular, for each state, we formulate 4 features, i.e., radio, computing, storage, and arrival event, as the input of the deep neural network. In this way, the training process is more efficient as all aspects of states are taken into account. The algorithm then updates the neural network by minimizing the following lost functions~\cite{Mnih2015Human}.
\begin{equation}
\label{lossfunction}
\begin{aligned}
L_i(\theta_i)\!\!=\!\!\mathbb{E}_{(\mathbf{s},a_{\mathbf{s}},r,\mathbf{s'})\sim U(D)}\!\!\bigg[ \!\!\bigg( r& \!+\!\gamma\mathcal{Q}(\mathbf{s'},\!\argmax_{a_{\mathbf{s'}}}\hat{\mathcal{Q}}(\mathbf{s'}, a_{\mathbf{s'}};\theta);\theta^-\!)\!)\\
&- \mathcal{Q}(\mathbf{s},a_{\mathbf{s}};\theta_i)\bigg)^2\bigg],
\end{aligned}
\end{equation}
where $\gamma$ is the discount factor, $\theta_i$ are the parameters of the Q-networks at episode $i$ and $\theta_i^-$ are the parameters of the target network, i.e., $\hat{\mathcal{Q}}$. $\theta_i$ and $\theta_i^-$ are used to compute the target at episode $i$.

Differentiating the loss function in~(\ref{lossfunction}) with respect to the parameters of the neural networks, we have the following gradient:
\begin{equation}
\label{gradient_loss}
\begin{aligned}
\nabla_{\theta_i}L(\theta_i)&\!\!=\!\!\mathbb{E}_{(\mathbf{s},a_{\mathbf{s}},r,\mathbf{s'})}\!\!\bigg[\!\!\bigg(r +\gamma\mathcal{Q}(\mathbf{s'},\argmax_{a_{\mathbf{s'}}}\hat{\mathcal{Q}}(\mathbf{s'}, a_{\mathbf{s'}};\theta);\theta^-))\\
&-\mathcal{Q}(\mathbf{s},a_{\mathbf{s}};\theta_i)\nabla_{\theta_i}\mathcal{Q}(\mathbf{s},a_{\mathbf{s}};\theta_i)\bigg)\bigg].
\end{aligned}
\end{equation}
From~(\ref{gradient_loss}), the loss function in~(\ref{lossfunction}) can be minimized by the \textit{Stochastic Gradient Descent} algorithm~\cite{Goodfellow2016Deep} that is the engine of most deep learning algorithms. Specifically, stochastic gradient descent is an extension of the gradient descent algorithm which is commonly used in machine learning. In general, the cost function used by a machine learning algorithm is decayed by a sum over training examples of some per-example loss function. For example, the negative conditional log-likelihood of the training data can be formulated as:
\begin{equation}
\begin{aligned}
J(\theta) &= \mathbb{E}_{(\mathbf{s},a_{\mathbf{s}},r,\mathbf{s'}) \sim U(D)}L\big((\mathbf{s},a_{\mathbf{s}},r,\mathbf{s'}),\theta\big)\\
&=\frac{1}{D}\sum_{i=1}^{D}L\big((\mathbf{s},a_{\mathbf{s}},r,\mathbf{s'})^{(i)}, \theta\big).
\end{aligned}
\end{equation}
For these additive cost function, gradient descent requires computing as follows:
\begin{equation}
\label{eq:cost}
\nabla_\theta J(\theta) = \frac{1}{D}\sum_{i=1}^{D} \nabla_\theta L\big((\mathbf{s},a_{\mathbf{s}},r,\mathbf{s'})^{(i)}, \theta\big).
\end{equation}
The computational cost for operation in~(\ref{eq:cost}) is $O(D)$. Thus, as the size $D$ of the replay memory is increased, the time to take a single gradient step becomes prohibitively long. As a result, the stochastic gradient descent technique is used in this paper. The core idea of using stochastic gradient descent is that the gradient is an expectation. Obviously, the expectation can be approximately estimated by using a small set of samples. In particular, we can uniformly sample a mini-batch of experiences from the replay memory at each step of the algorithm. Typically, the mini-batch size can be set to be relative small number of experiences, e.g., from 1 to a few hundreds. As such, the training time is significantly fast. The estimate of the gradient under the stochastic gradient descent is then rewritten as follows:
\begin{equation}
g=\frac{1}{D'} \nabla_\theta\sum_{i=1}^{D'}L\big((\mathbf{s},a_{\mathbf{s}},r,\mathbf{s'})^{(i)}, \theta \big),
\end{equation}
where $D'$ is the mini-batch size. The stochastic gradient descent algorithm then follows the estimated gradient downhill as in~(\ref{eq:downhill}).
\begin{equation}
\label{eq:downhill}
\theta \leftarrow \theta- \nu g,
\end{equation}
where $\nu$ is the learning rate of the algorithm.

The target network parameters $\theta_i^-$ are only updated with the Q-network parameters $\theta_i$ every $C$ steps and are remained fixed between individual updates. It is worth noting that the training process of the deep double Q-learning algorithm is different from the training process in supervised learning by updating the network parameters using previous experiences in an online manner. %The flowchart of the deep double Q-learning algorithm is shown in Fig.~\ref{Fig.flowchart}
\begin{figure}[h]
	\centering
	\includegraphics[scale=0.13]{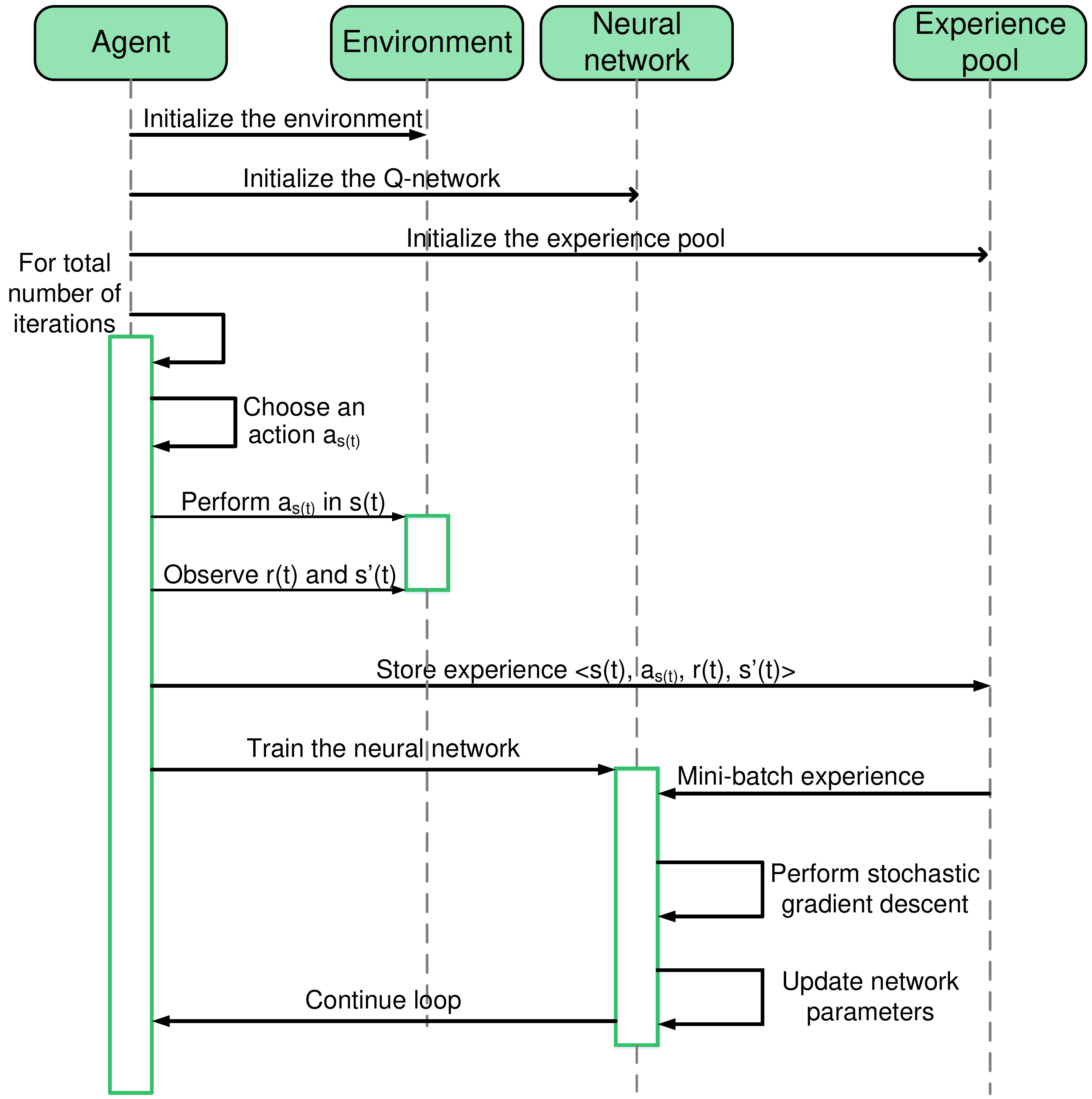}
	\caption{Flow chart of the deep double Q-learning algorithm.}
	\label{Fig.flowchart}
\end{figure}

%======================================================
%======================================================
\subsection{Deep Dueling Network}

Due to the overestimation of optimizers, the convergence rates of the deep Q-learning and deep double Q-learning learning algorithms are still limited, especially in large-scale systems~\cite{Wang2015Dueling}. Therefore, we propose the novel network slicing framework using the deep dueling algorithm~\cite{Wang2015Dueling}, which is also originally developed by Google DeepMind in 2016, to further improve the system's convergence speed. The key idea making the deep dueling superior to conventional approaches is its novel neural network architecture. In this neural network, instead of estimating the action-value function, i.e., Q-function, the values of states and advantages of actions{\footnote{The value function represents how good it is for the system to be in a given state. The advantage function is used to measure the importance of a certain action compared with others \cite{Wang2015Dueling}.}} are separately estimated by two sequences, i.e., streams, of fully connected layers. The values and advantages are combined at the output layer as shown in Fig.~\ref{Fig.deepduelingqlearning}. The reason behind this architecture is that in many states it is unnecessary to estimate the value of corresponding actions as the choice of these actions has no repercussion on what happens~\cite{Wang2015Dueling}. In this way, the deep dueling algorithm can achieve more robust estimates of state value, thereby significantly improving its convergence rate as well as stability.

\begin{figure}[h]
	\centering
	\includegraphics[scale=0.22]{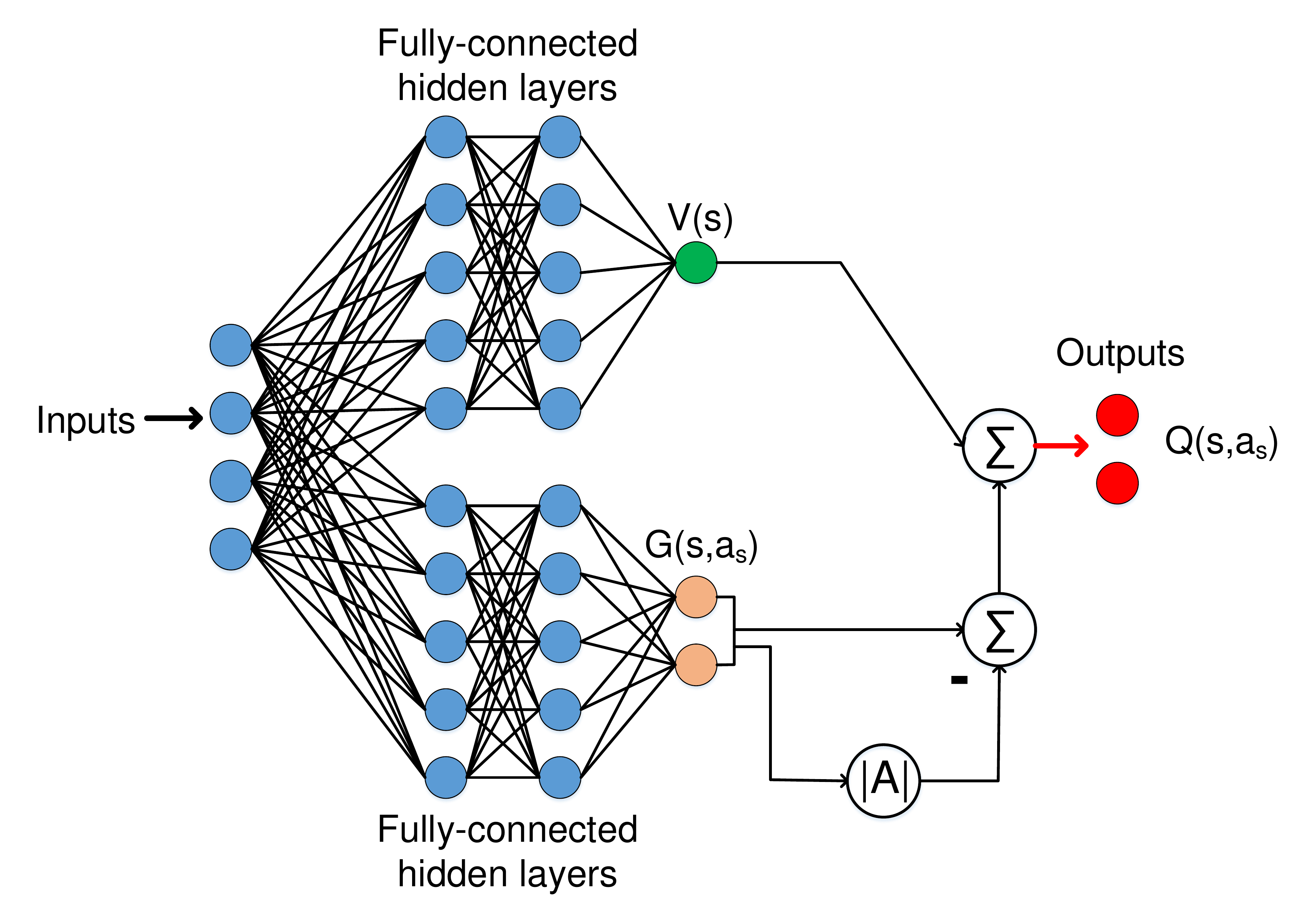}
	\caption{Deep dueling model.}
	\label{Fig.deepduelingqlearning}
\end{figure}

Recall that given a stochastic policy $\pi$, the values of state-action pair $(\mathbf{s},a_{\mathbf{s}})$ and state $\mathbf{s}$ are expressed as:
\begin{equation}
\begin{aligned}
\mathcal{Q}^{\pi}(\mathbf{s},a_{\mathbf{s}}) &= \mathbb{E}\big[\mathcal{R}_t|\mathbf{s}_t=\mathbf{s}, a_{\mathbf{s}_t}=a_{\mathbf{s}},\pi\big] \mbox{ and }\\
\mathcal{V}_{\pi}(\mathbf{s}) &= \mathbb{E}_{a_{\mathbf{s}} \sim \pi(\mathbf{s})}\big[\mathcal{Q}^{\pi}(\mathbf{s},a_{\mathbf{s}})\big].
\end{aligned}
\end{equation}
The advantage function of actions can be computed as:
\begin{equation}
\mathcal{G}^{\pi}(\mathbf{s},a_{\mathbf{s}}) =  \mathcal{Q}^{\pi}(\mathbf{s},a_{\mathbf{s}}) - \mathcal{V}^{\pi}(\mathbf{s},a_{\mathbf{s}}).
\end{equation}
Specifically, the value function $\mathcal{V}$ corresponds to how good it is to be in a particular state $\mathbf{s}$. The state-action pair, i.e., Q-function, measures the value of selecting action $a_{\mathbf{s}}$ in state $\mathbf{s}$. The advantage function decouples the state value from Q-function to obtain a relative measure of the importance of each action. It is important to note that $\mathbb{E}_{a_{\mathbf{s}} \sim \pi(\mathbf{s})}\big[\mathcal{G}^{\pi}(\mathbf{s},a_{\mathbf{s}})\big] = 0$. In addition, given a deterministic policy $a_{\mathbf{s}}^*= \argmax_{a_{\mathbf{s}} \in \mathcal{A}}\mathcal{Q}(\mathbf{s}, a_{\mathbf{s}})$, we have $\mathcal{Q}(\mathbf{s}, a_{\mathbf{s}}^*)=\mathcal{V}(\mathbf{s})$, and hence $\mathcal{G}(\mathbf{s}, a_{\mathbf{s}}^*)=0$. 

To estimate values of $\mathcal{V}$ and $\mathcal{G}$ functions, we use a dueling neural network in which one stream of fully-connected layers outputs a scalar $\mathcal{V}(\mathbf{s};\beta)$ and the other stream outputs an $|\mathcal{A}|$-dimensional vector $\mathcal{G}(\mathbf{s}, a_{\mathbf{s}};\alpha)$ with $\alpha$ and $\beta$ are the parameters of fully-connected layers. These two streams are then combined to obtain the Q-function by~(\ref{combined}).
\begin{equation}
\label{combined}
\mathcal{Q}(\mathbf{s}, a_{\mathbf{s}};\alpha, \beta) = \mathcal{V}(\mathbf{s};\beta) + \mathcal{G}(\mathbf{s}, a_{\mathbf{s}};\alpha).
\end{equation}
However, $\mathcal{Q}(\mathbf{s}, a_{\mathbf{s}};\alpha, \beta)$ is only a parameterized estimate of the true Q-function. Moreover, given $Q$, we cannot obtain $\mathcal{V}$ and $\mathcal{G}$ uniquely. Therefore, (\ref{combined}) is unidentifiable resulting in poor performance. To address this issue, we let the combining module of the network implement the following mapping:
\begin{equation}
\label{ouput_max}
\mathcal{Q}(\mathbf{s},a_{\mathbf{s}};\alpha,\beta) = \mathcal{V}(\mathbf{s};\beta) + \big(\mathcal{G}(\mathbf{s},a_{\mathbf{s}};\alpha)-\max_{a_{\mathbf{s}} \in \mathcal{A}}\mathcal{G}(\mathbf{s},a_{\mathbf{s}};\alpha)\big).
\end{equation}
By doing this, the advantage function estimator has zero advantage when choosing action. Intuitively, given $a_{\mathbf{s}}^*=\argmax_{a_{\mathbf{s}}\in \mathcal{A}} \mathcal{Q}(\mathbf{s},a_{\mathbf{s}};\alpha,\beta)=\argmax_{a_{\mathbf{s}} \in \mathcal{A}}\mathcal{G}(\mathbf{s},a_{\mathbf{s}};\alpha)$, we have $\mathcal{Q}(\mathbf{s}, a_{\mathbf{s}}^*;\alpha, \beta)=\mathcal{V}(\mathbf{s};\beta)$. (~\ref{ouput_max}) can be transformed into a simple form by replacing the max operator with an average as in~(\ref{output_average}).
\begin{equation}
\label{output_average}
\mathcal{Q}(\mathbf{s},a_{\mathbf{s}};\alpha,\beta) = \mathcal{V}(\mathbf{s};\beta) + \big(\mathcal{G}(\mathbf{s},a_{\mathbf{s}};\alpha)- \frac{1}{|\mathcal{A}|}\sum_{a_{\mathbf{s}}}^{}\mathcal{G}(\mathbf{s}, a_{\mathbf{s}};\alpha)\big).
\end{equation}
\begin{algorithm}[h]
	\caption{Deep Dueling Network Based Resources Slicing Algorithm}
	\label{deepduelingqlearning}
	\begin{algorithmic}[1]
		\State Initialize replay memory to capacity $D$.
		\State Initialize the primary network $\mathcal{Q}$ including two fully-connected layers with random weights $\alpha$ and $\beta$.
		\State Initialize the target network $\hat{\mathcal{Q}}$ as a copy of the primary Q-network with weights $\alpha^- = \alpha$ and $\beta^- = \beta$.
		\For{\textit{episode=1 to T}}
		\State Base on the epsilon-greedy algorithm, with probability $\epsilon$ select a random action $a_{\mathbf{s}_t}$, otherwise
		\State select $a_{\mathbf{s}_t}=\argmax \mathcal{Q}^*(\mathbf{s}_t, a_{\mathbf{s}_t}; \alpha, \beta)$
		\State Perform action $a_{\mathbf{s}_t}$ and observe reward $r_t$ and next state $\mathbf{s}_{t+1}$
		\State Store transition $(\mathbf{s}_t, a_{\mathbf{s}_t}, r_t, \mathbf{s}_{t+1})$ in the replay memory
		\State Sample random minibatch of transitions $(\mathbf{s}_j, a_{\mathbf{s}_j}, r_j, \mathbf{s}_{j+1})$ from the replay memory
		\State Combine the value function and advantage functions as follows:
		\begin{equation}
		\begin{aligned}
		\mathcal{Q}(\mathbf{s}_j,a_{\mathbf{s}_j};\alpha,\beta) = \mathcal{V}(\mathbf{s}_j;\beta) &+ \big(\mathcal{G}(\mathbf{s}_j,a_{\mathbf{s}_j};\alpha)\\
		&- \frac{1}{|\mathcal{A}|}\sum_{a_{\mathbf{s}_j}}^{}\mathcal{G}(\mathbf{s}_j, a_{\mathbf{s}_j};\alpha)\big).
		\end{aligned}
		\end{equation}
		\State $y_j=r_j+\gamma\max_{a_{\mathbf{s}_{j+1}}}\hat{\mathcal{Q}}(s_{j+1},a_{\mathbf{s}_{j+1}}; \alpha^-, \beta^-)$
		%		\EndIf
		\State Perform a gradient descent step on $(y_j-\mathcal{Q}(\mathbf{s}_j, a_{\mathbf{s}_j}; \alpha, \beta))^2$
		\State Every $C$ steps reset $\hat{\mathcal{Q}} = \mathcal{Q}$
		\EndFor
	\end{algorithmic}
\end{algorithm}

Based on~(\ref{output_average}) and the advantages of the deep reinforcement learning, the details of the deep dueling algorithm used in our proposed approach are shown in Algorithm~\ref{deepduelingqlearning}. It is important to note that~(\ref{output_average}) is viewed and implemented as a part of the network and not as a separate algorithmic step~\cite{Wang2015Dueling}. In addition, $\mathcal{V}(\mathbf{s};\beta)$ and $\mathcal{G}(\mathbf{s},a_{\mathbf{s}};\alpha)$ are estimated automatically without any extra supervision or modifications in the algorithm.

%====================================================================================
%====================================================================================	
\section{Performance Evaluation}
\label{sec:evaluation}

%======================================================
%======================================================
\subsection{Parameter Setting}

We perform the simulations using TensorFlow~\cite{TensorFlow} to evaluate the performance of the proposed solutions under different parameter settings. We consider three common classes of slices, i.e., utilities (class-1), automotive (class-2), and manufacturing (class-3). Unless otherwise stated, the arrival rates $\lambda_c$ of requests from class-1, class-2, and class-3 are set at 12 requests/hour, 8 requests/hour, and 10 requests/hour, respectively. The completion rates $\mu_c$ of requests from class-1, class-2, and class-3 are set at 3 requests/hour. The immediate reward $r_c$ for each accepted request from class-1, class-2, and class-3 are 1, 2, and 4, respectively. These parameters will be varied later to evaluate the impacts of the immediate reward on the decisions of the RMO. Each slice request requires 1 GB of storage resources, 2 CPUs for computing, and 100 Mbps of radio resources~\cite{SattarOptimal}. Importantly, the architecture of the deep neural network requires thoughtful design as it greatly affects the performance of the algorithm. Intuitively, increasing the number of hidden layers will increase the complexity of the algorithm. However, when the number of hidden layers is very small, the algorithm may not converge to the optimal policy. Similarly, when the size of hidden layers and mini-batch size are large, the algorithm will need more time to estimate the Q-function. In our experiment, we choose these parameters based on common settings in the literature~\cite{Mnih2015Human,Wang2015Dueling}. In particular, for the deep Q-learning and deep double Q-learning algorithms, two fully-connected hidden layers are implemented together with input and output layers. For the deep dueling algorithm, the neural network is divided into two streams~\cite{Wang2015Dueling}. Each stream consists of a hidden layer connected to the input and output layers. The size of the hidden layers is 64. The mini-batch size is set at 64. Both the Q-learning algorithm and the deep reinforcement learning algorithms use $\epsilon$-greedy algorithm with the initial value of $\epsilon$ is 1, and its final value is 0.1~\cite{Watkins1992QLearning,Chen2018Reinforcement}. The maximum size of the experience replay buffer is 10,000, and the target Q-network is updated every 1,000 iterations~\cite{Mnih2015Human},~\cite{Goodfellow2016Deep}.

%\newpage
%======================================================
%======================================================
\subsection{Simulation Results}
\label{sec:evaluationB}

%==========================
%==========================
\subsubsection{Performance Evaluation}

\paragraph{Comparison to Existing Network Slicing Solutions}
\begin{figure}[!]
	\centering
	\includegraphics[scale=0.27]{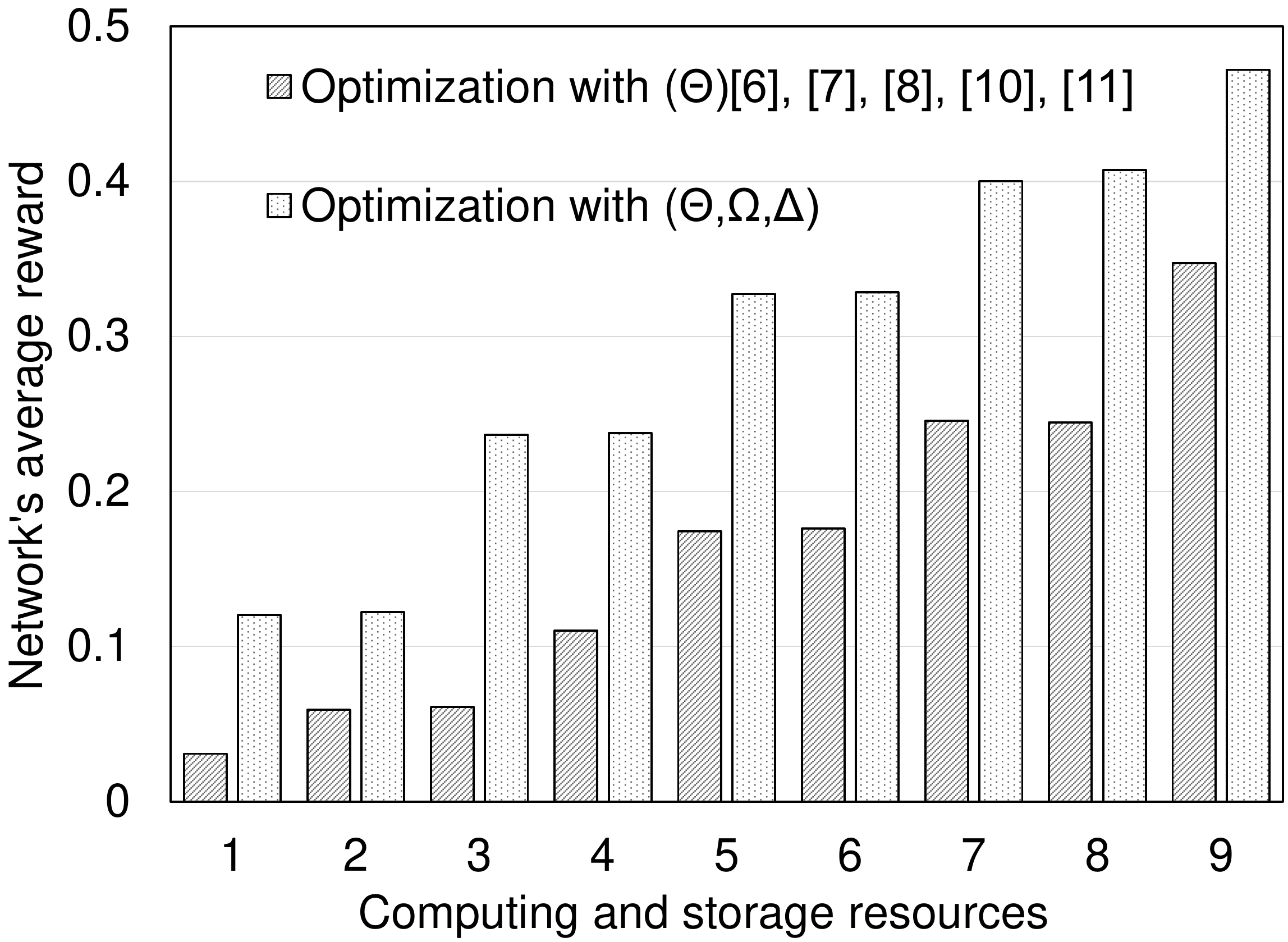}
	\caption{The average reward when optimizing with one resource and three resources.}
	\label{Fig.compare}
\end{figure}
\begin{figure}[!]
	\centering
	\includegraphics[scale=0.33]{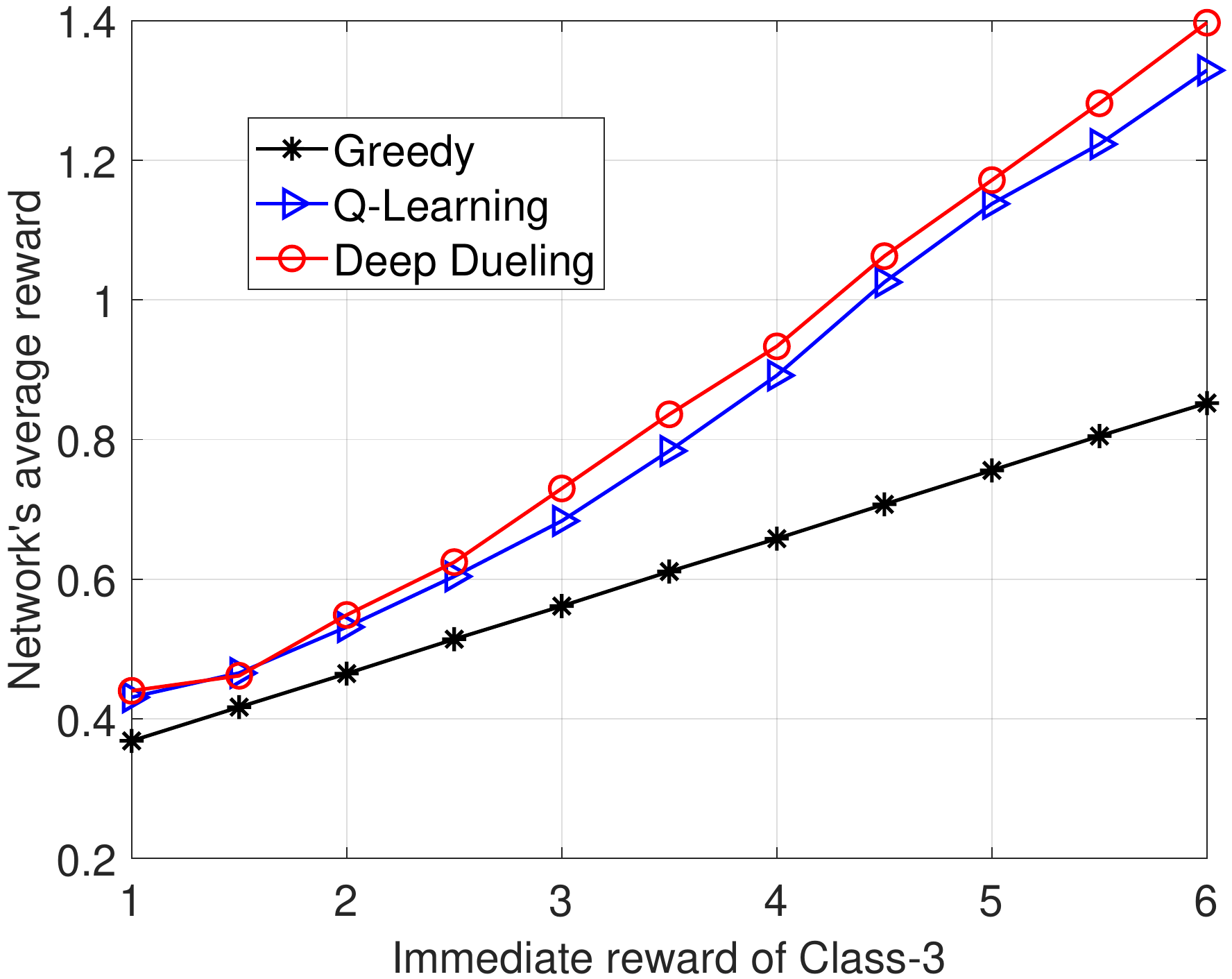}
	\caption{The average reward of the system when the immediate reward of class-3 is varied.}
	\label{Fig.reward_vary_rew}
\end{figure}
As mentioned, most existing works, e.g., \cite{Jiang2016Network, Soliman2016QoS, Sciancalepore, Bega2017Optimising, Aijaz2017Hap} optimized slicing for only the radio resource. In practice, besides the radio resource, both computing and storage resource should also be accounted for while orchestrating slices. This makes existing solutions sub-optimal. In this section, we set the maximum radio resources at 500 Mbps. Each request requires 50 Mbps for radio access, 2 CPUs for computing, and 2 GB of storage resources. The computing and storage resources are then varied from 1 CPU to 9 CPUs and 1 GB to 9 GB, respectively. Fig.~\ref{Fig.compare} shows the average reward of the system obtained by the Q-learning algorithm for the case with three resources are taken into account (as in our considered system model) and for the case with only radio resource as considered in~\cite{Jiang2016Network, Soliman2016QoS, Sciancalepore, Bega2017Optimising, Aijaz2017Hap}. As can be observed, when the computing and storage resources increase, the average reward is increased as more slice requests are accepted. However, the average reward of our approach (taking all radio, computing, and storage resources into account) is significantly higher than those of other solutions in the literature, especially when the amount of computing and storage resources are small. This is due to the fact that slices not only request radio resources to ensure the bandwidth for connections but also computing and storage resources to fulfill the requirements of different services.

\paragraph{Average Reward and Network Performance}
\begin{figure*}[!]
	\centering
	\begin{subfigure}[b]{0.28\textwidth}
		\centering
		\includegraphics[scale=0.28]{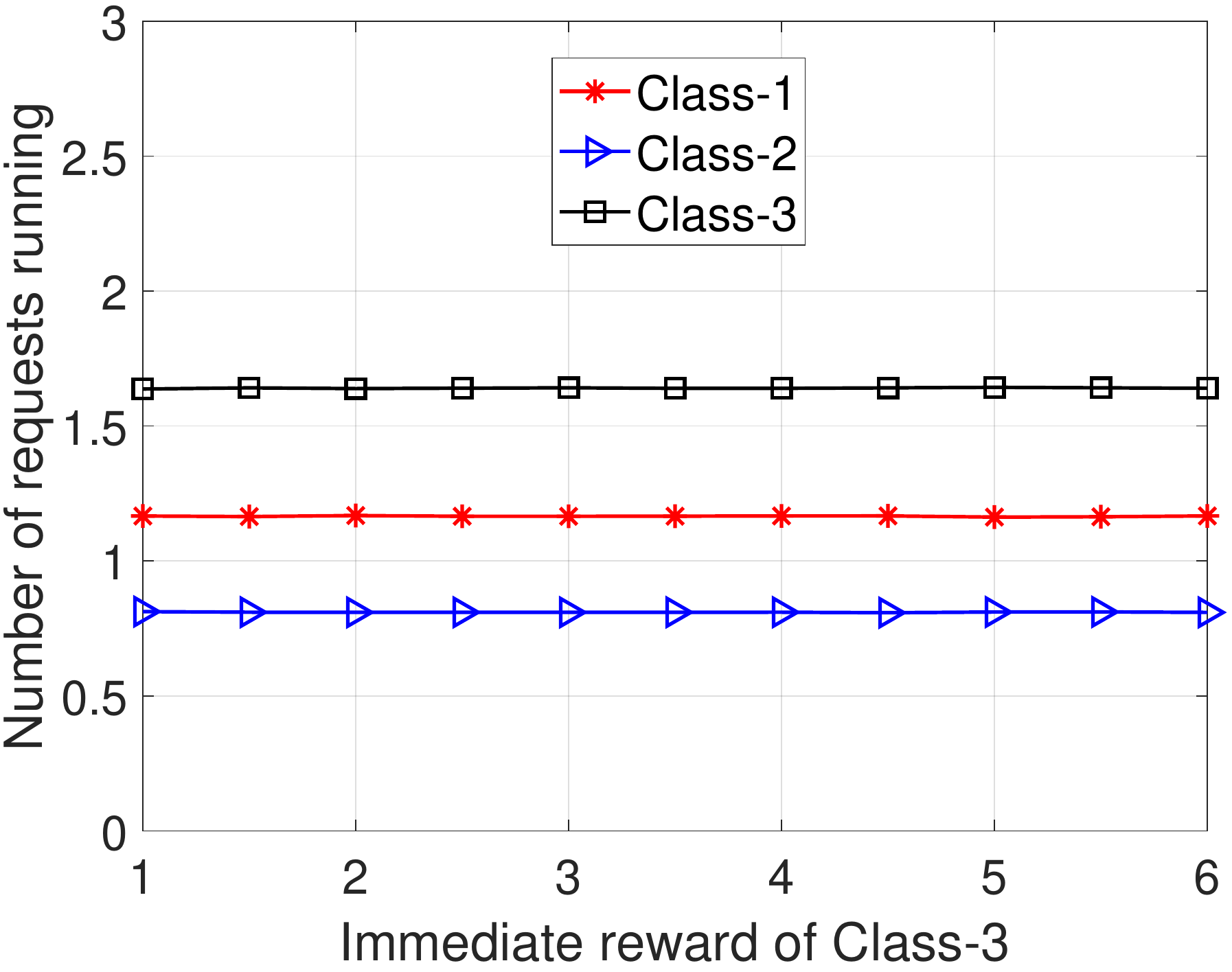}
		\caption{}
	\end{subfigure}%
	~ 
	\begin{subfigure}[b]{0.28\textwidth}
		\centering
		\includegraphics[scale=0.28]{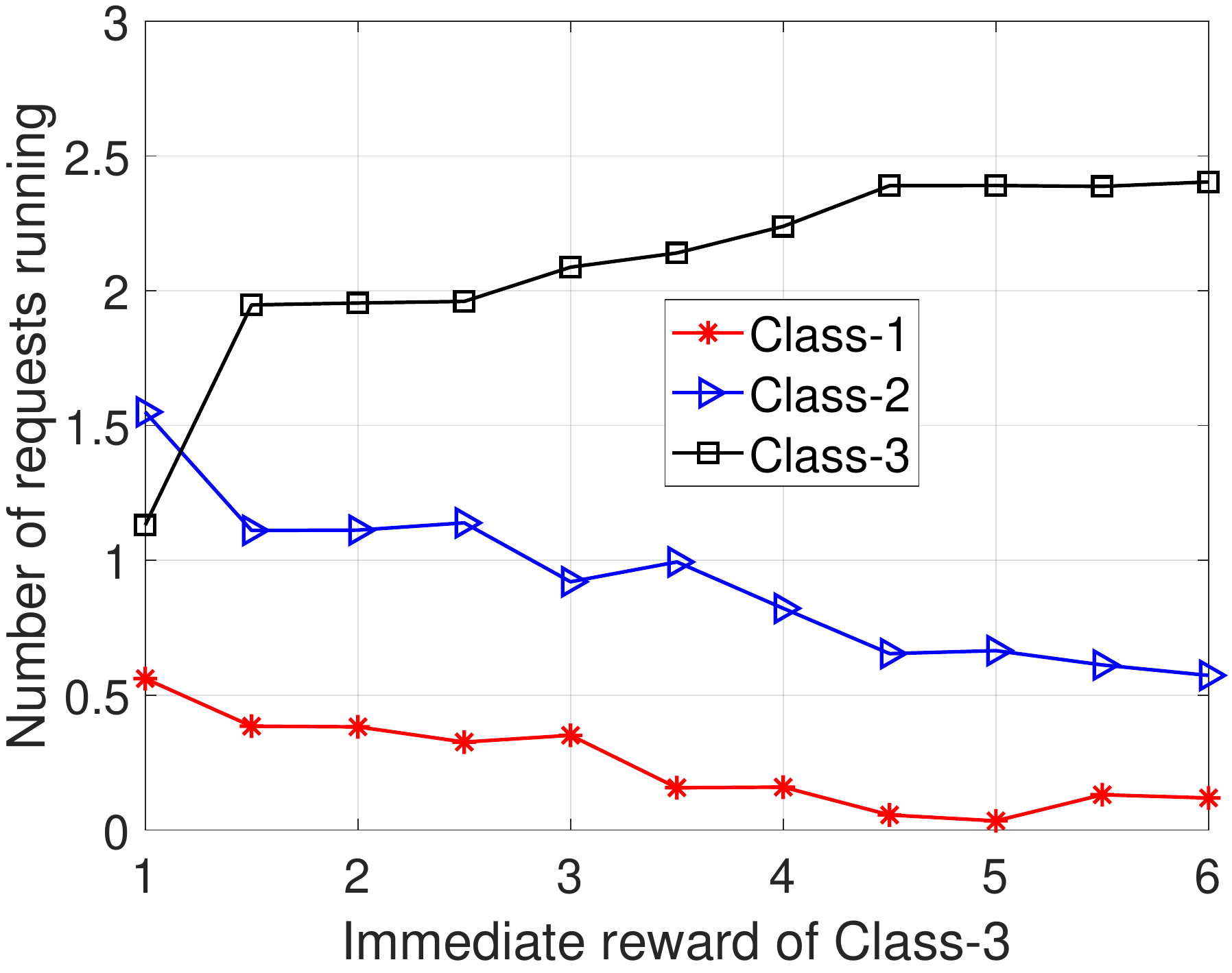}
		\caption{}
	\end{subfigure}%
	~ 
	\begin{subfigure}[b]{0.28\textwidth}
		\centering
		\includegraphics[scale=0.28]{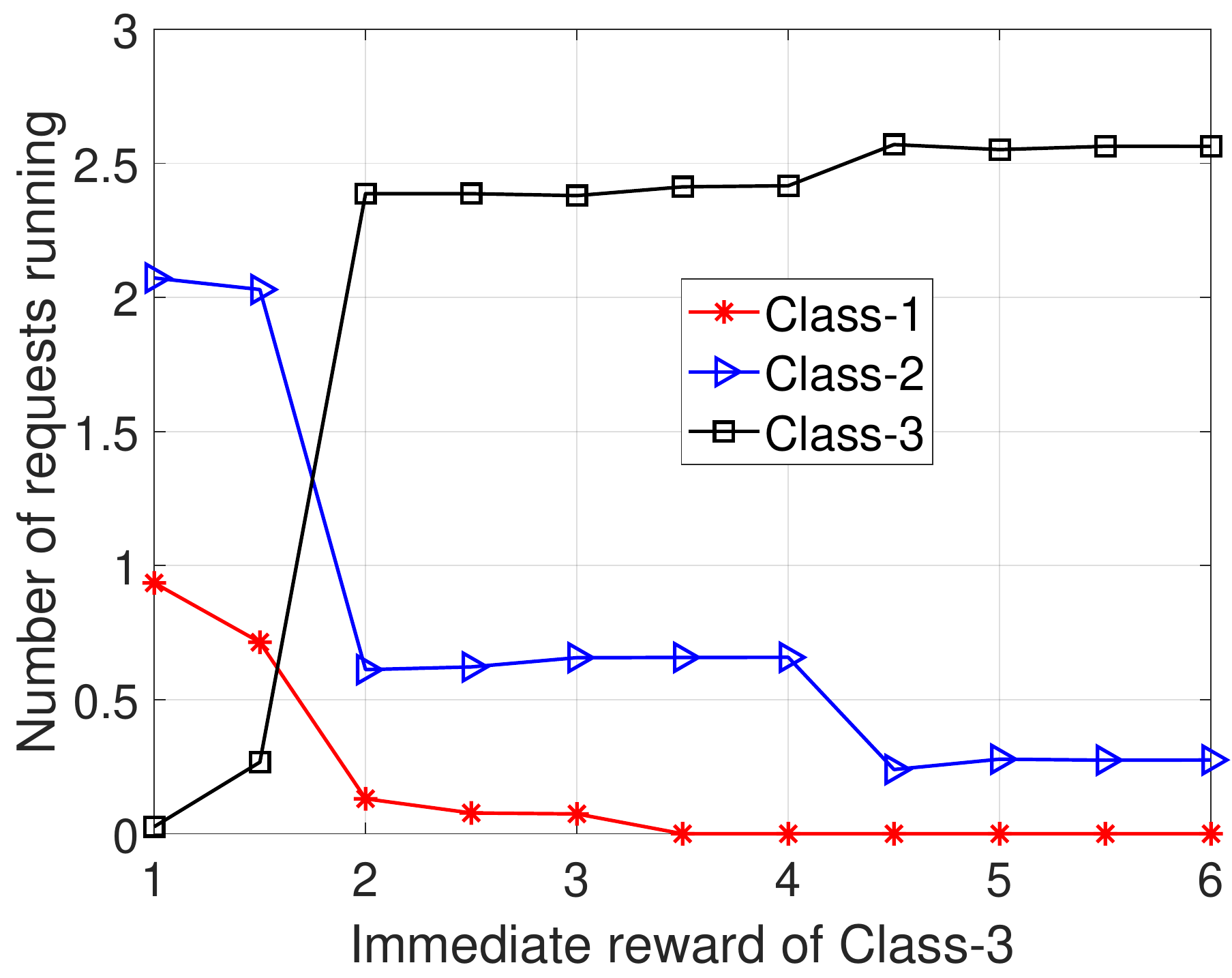}
		\caption{}
	\end{subfigure}
	\caption{The number of request running in the system of (a) greedy algorithm, (b) Q-learning algorithm, and (c) deep dueling algorithm when the immediate reward of class-3 is varied.} 
	\label{fig:vary_reward_running}
\end{figure*}

Next, we compare the performance of the proposed solution, i.e., deep dueling algorithm, with other methods, i.e., Q-learning~\cite{Bega2017Optimising} and greedy algorithms~\cite{Aijaz2017Hap},~\cite{Han2018Sice}, in terms of average reward and the number of requests running in the system. For a small-size system (the maximum radio, computing, and storage resources are set at 400 Mbps, 8 CPUs, and 4 GB, respectively), Fig.~\ref{Fig.reward_vary_rew} shows the average reward of the system obtained by three algorithms while varying the reward of slices from class-3 from 1 to 6. As can be seen, with the increasing of the reward of slices from class-3, the average reward of the system is increased. However, the average reward obtained by the reinforcement learning algorithms, i.e., deep dueling and Q-learning, is significantly higher than that of the greedy algorithm. This is due to the fact that the proposed reinforcement learning approaches reserve resources for coming requests that may have high rewards, while the greedy algorithm accepts slices based on the available resource of the system as shown in Fig.~\ref{fig:vary_reward_running}. It is worth noting that the achieved reward of the Q-learning algorithm is not as good as the reward obtained by the deep dueling algorithm even with small-size scenarios. This is because that the Q-learning algorithm has a slow convergence rate due to the curse-of-dimensionality problem. This observation is more pronounced when we later increase the size of the system.

As observed in Fig.~\ref{fig:vary_reward_running}, the number of requests running in the systems under the greedy algorithm remains the same when the immediate reward of slices from class-3 is varied. The reason is that the greedy algorithm does not consider the immediate reward of slice requests into account. In other words, upon receiving a slice request, the greedy algorithm will accept this request if the available resources of the infrastructure satisfy the slice service demands. In contrast, for the reinforcement learning algorithms, the immediate reward is also an essential factor to make optimal decisions. In particular, when the immediate reward of slice requests from class-3 increases, the algorithms are likely to reject the slice requests from classes which have lower immediate rewards, i.e., slice requests from class-1. For example, when the immediate reward of slice request from class-3 is 6, the number of requests from class-1, whose immediate reward is 1, approaches 0.

To observe the performance of the proposed solutions when the state space of the system is large, we increase the radio, computing, and storage resources to 2 Gbps, 40 CPUs, and 20 GB, respectively.
\begin{figure}[!]
	\centering
	\includegraphics[scale=0.33]{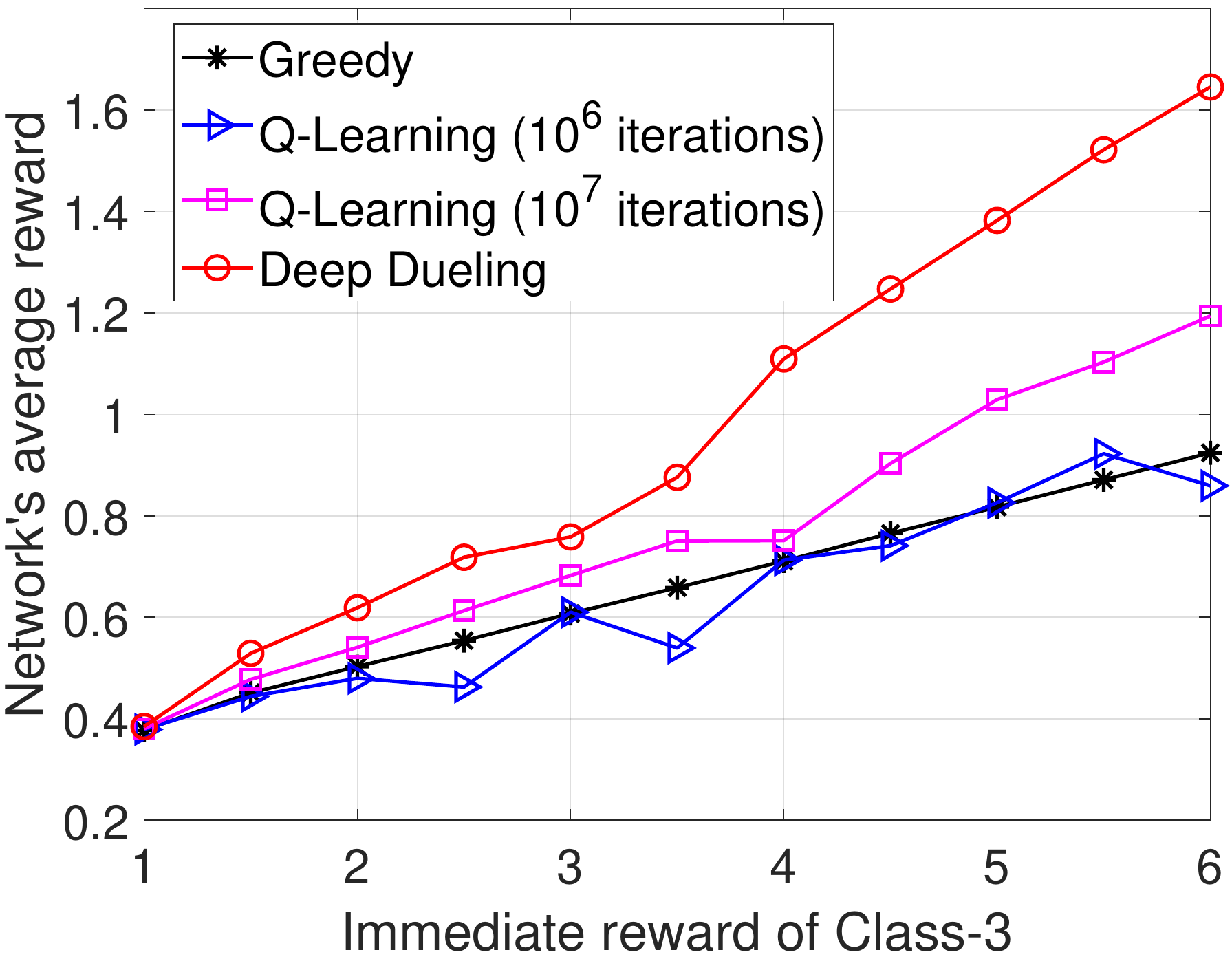}
	\caption{The average reward of the system when the immediate reward of class-3 is varied.}
	\label{Fig.reward_vary_rew_20}
\end{figure}
The arrival rate of requests from class-1 is 48 requests/hour, from class-2 is 32 requests/hour, and from class-3 is 40 requests/hour. The completion rates from all classes are set at 2 requests/hour. Fig.~\ref{Fig.reward_vary_rew_20} shows that the average reward obtained by the deep dueling algorithm is much higher than those of the greedy and Q-learning algorithms. This is because of the slow convergence of the Q-learning algorithm to optimality. Specifically, with $10^6$ iterations, the performance of the Q-learning algorithm is just the same as that of the greedy algorithm. The performance of the Q-learning algorithm is improved with $10^7$ iterations, but it is still way inferior to that of the deep dueling algorithm. For this large system scenario with over 74,000 state-action pairs, on a laptop with Intel Core i7-7600U and 16GB RAM, the deep dueling algorithm just takes about 2 hours to finish 15,000 iterations and obtain the optimal policy. This is a very practical number compared with the Q-learning algorithm that cannot obtain the optimal policy within $10^7$ iterations (more than 15 hours). In practice, with specialized hardware and much more powerful computing resource (compared with our laptop) at the network provider (e.g., GPU cards from NVIDIA), the deep dueling algorithm should take much shorter than 2 hours to finish 15,000 iterations~\cite{Uz}. These results confirm that the Q-learning algorithm, despite its optimality, requires a much longer time to converge, compared with the deep dueling algorithm.

%\begin{figure*}[!]
%	\centering
%	\begin{subfigure}[b]{0.25\textwidth}
%		\centering
%		\includegraphics[scale=0.25]{Figures/combined_greedy_q1M.pdf}
%		\caption{}
%	\end{subfigure}%
%	~ 
%	\begin{subfigure}[b]{0.25\textwidth}
%		\centering
%		\includegraphics[scale=0.25]{Figures/combined_greedy_q10M}
%		\caption{}
%	\end{subfigure}%
%	~ 
%	\begin{subfigure}[b]{0.25\textwidth}
%		\centering
%		\includegraphics[scale=0.25]{Figures/combined_greedy_dueling}
%		\caption{}
%	\end{subfigure}
%	\caption{The number of request running in the system of (a) Q-learning algorithm ($10^6$ iterations), (b) Q-learning algorithm ($10^7$ iterations), and (c) deep dueling algorithm (20,000 iterations) when the immediate reward of class-3 is varied. The dash lines are results of the greedy algorithm.} 
%	\label{fig:vary_reward_running_20}
%\end{figure*}

\begin{figure*}[htbp]
	\begin{minipage}[t]{1.0\textwidth}
	\centering
	\begin{subfigure}[b]{0.3\textwidth}
		\centering
		\includegraphics[scale=0.28]{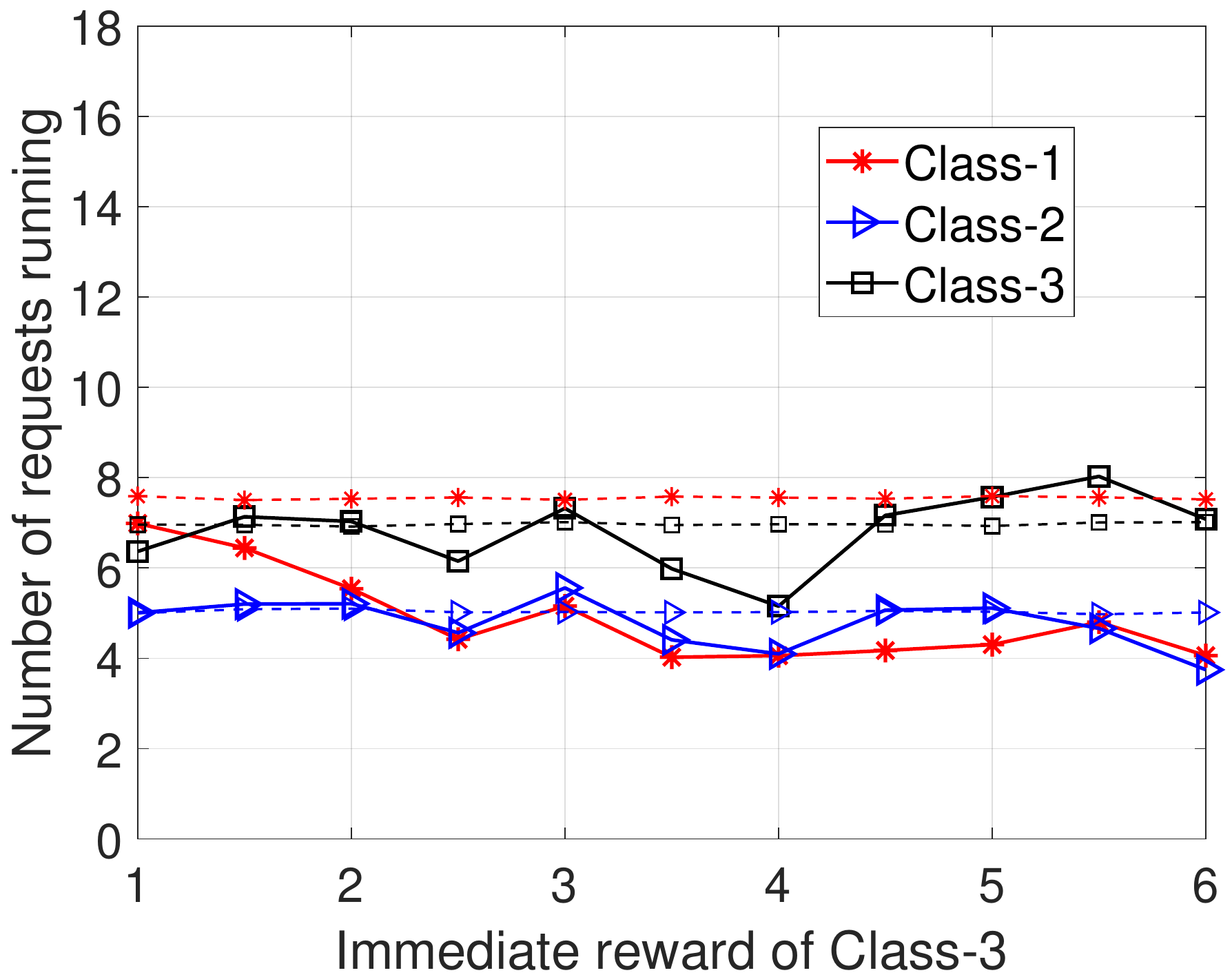}
		\caption{}
	\end{subfigure}%
	~ 
	\begin{subfigure}[b]{0.3\textwidth}
		\centering
		\includegraphics[scale=0.28]{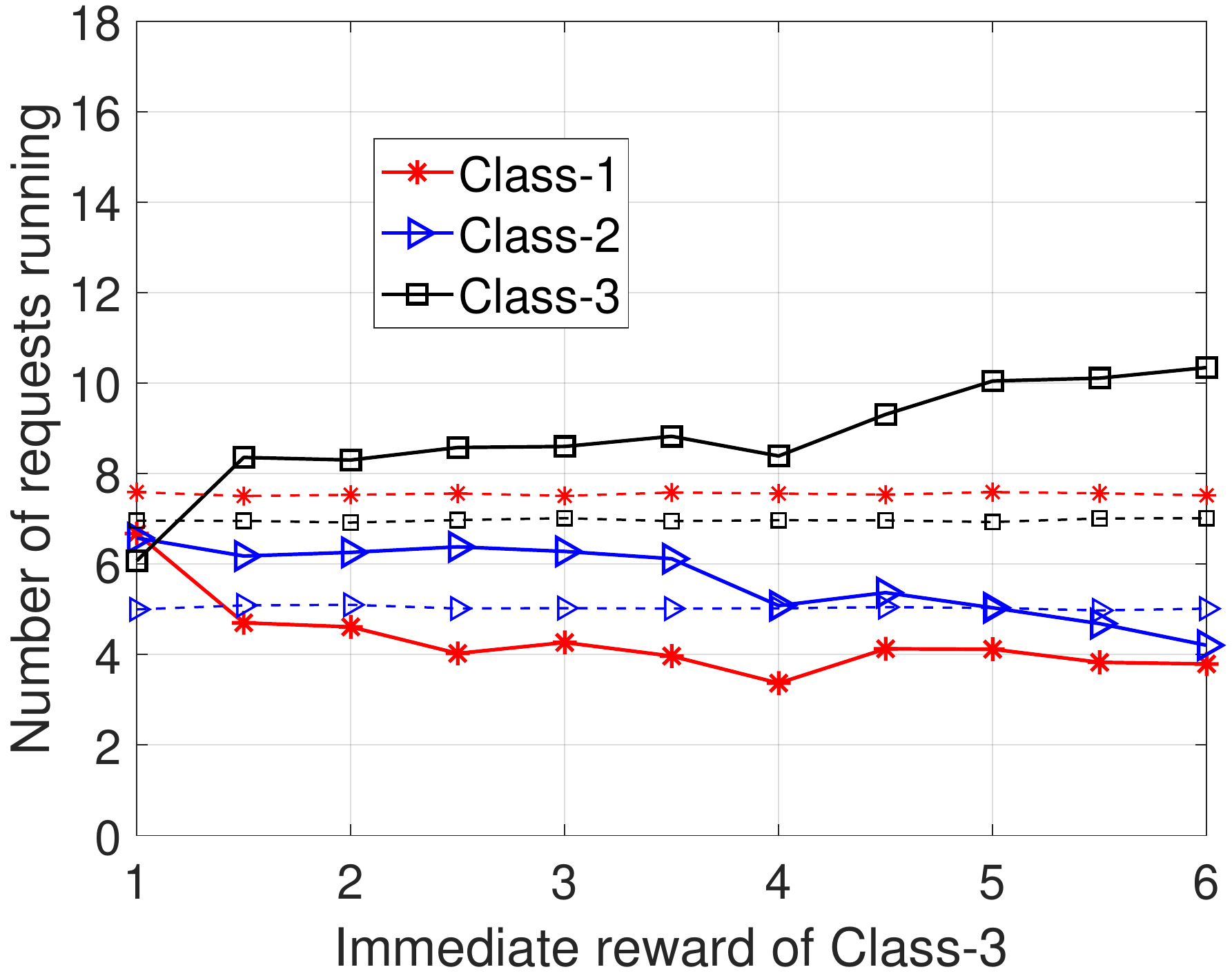}
		\caption{}
	\end{subfigure}%
	~ 
	\begin{subfigure}[b]{0.3\textwidth}
		\centering
		\includegraphics[scale=0.28]{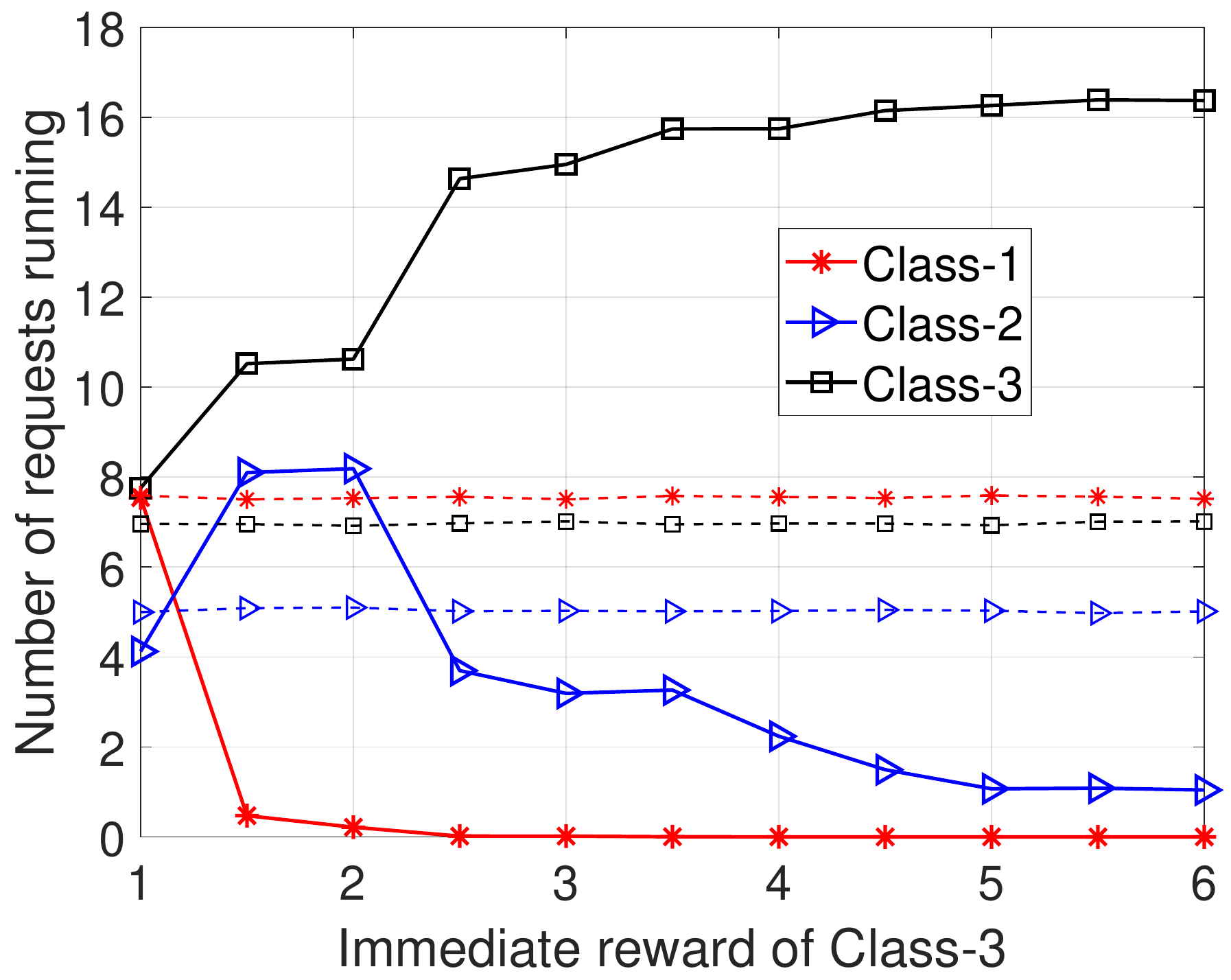}
		\caption{}
	\end{subfigure}
	\caption{The number of request running in the system of (a) Q-learning algorithm ($10^6$ iterations), (b) Q-learning algorithm ($10^7$ iterations), and (c) deep dueling algorithm (20,000 iterations) when the immediate reward of class-3 is varied. The dash lines are results of the greedy algorithm.} 
	\label{fig:vary_reward_running_20}
	\end{minipage}

	\begin{minipage}[t]{1.0\textwidth}
		\centering
		\begin{subfigure}[b]{0.3\textwidth}
			\centering
			\includegraphics[scale=0.28]{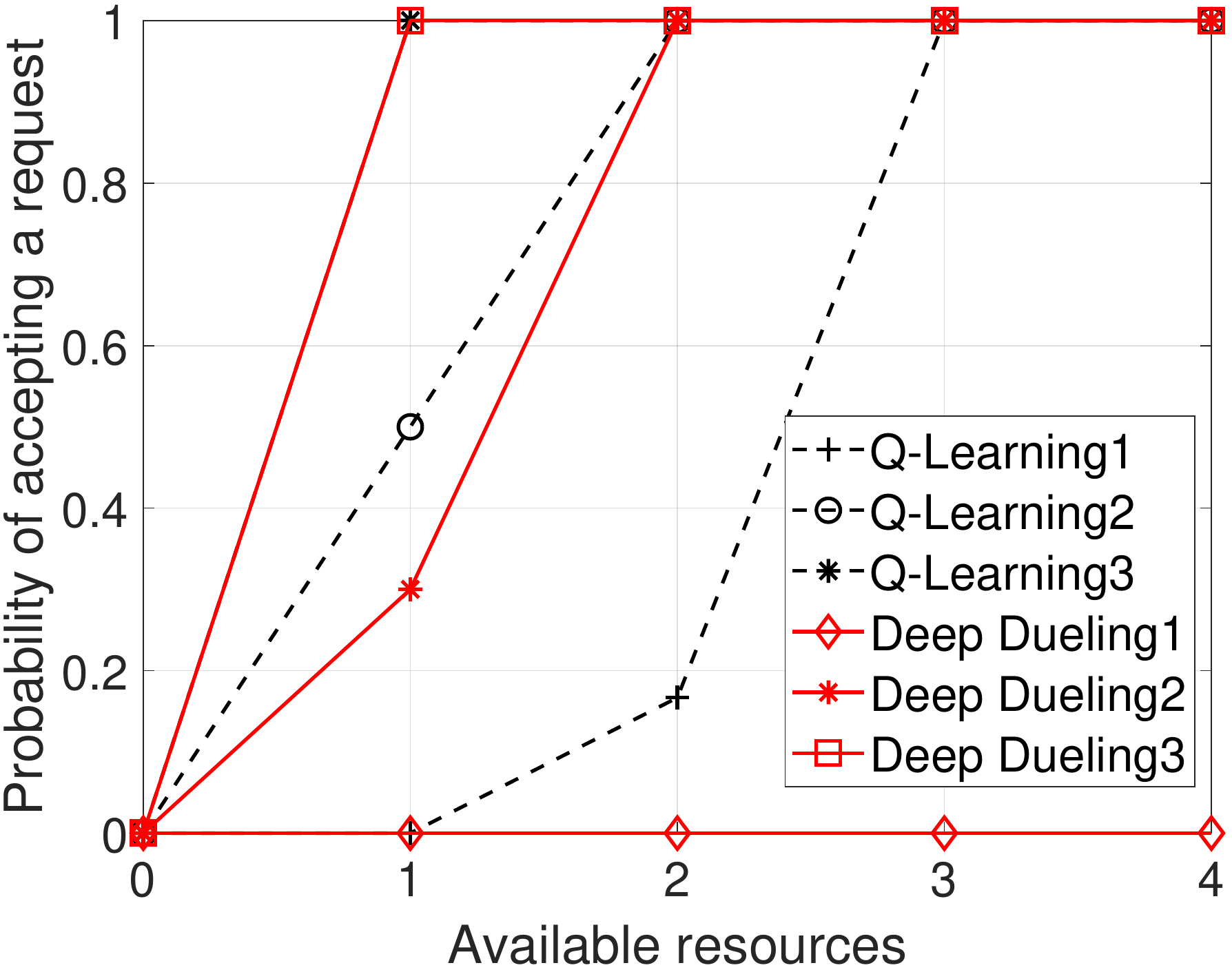}
			\caption{}
		\end{subfigure}%
		~ 
		\begin{subfigure}[b]{0.3\textwidth}
			\centering
			\includegraphics[scale=0.28]{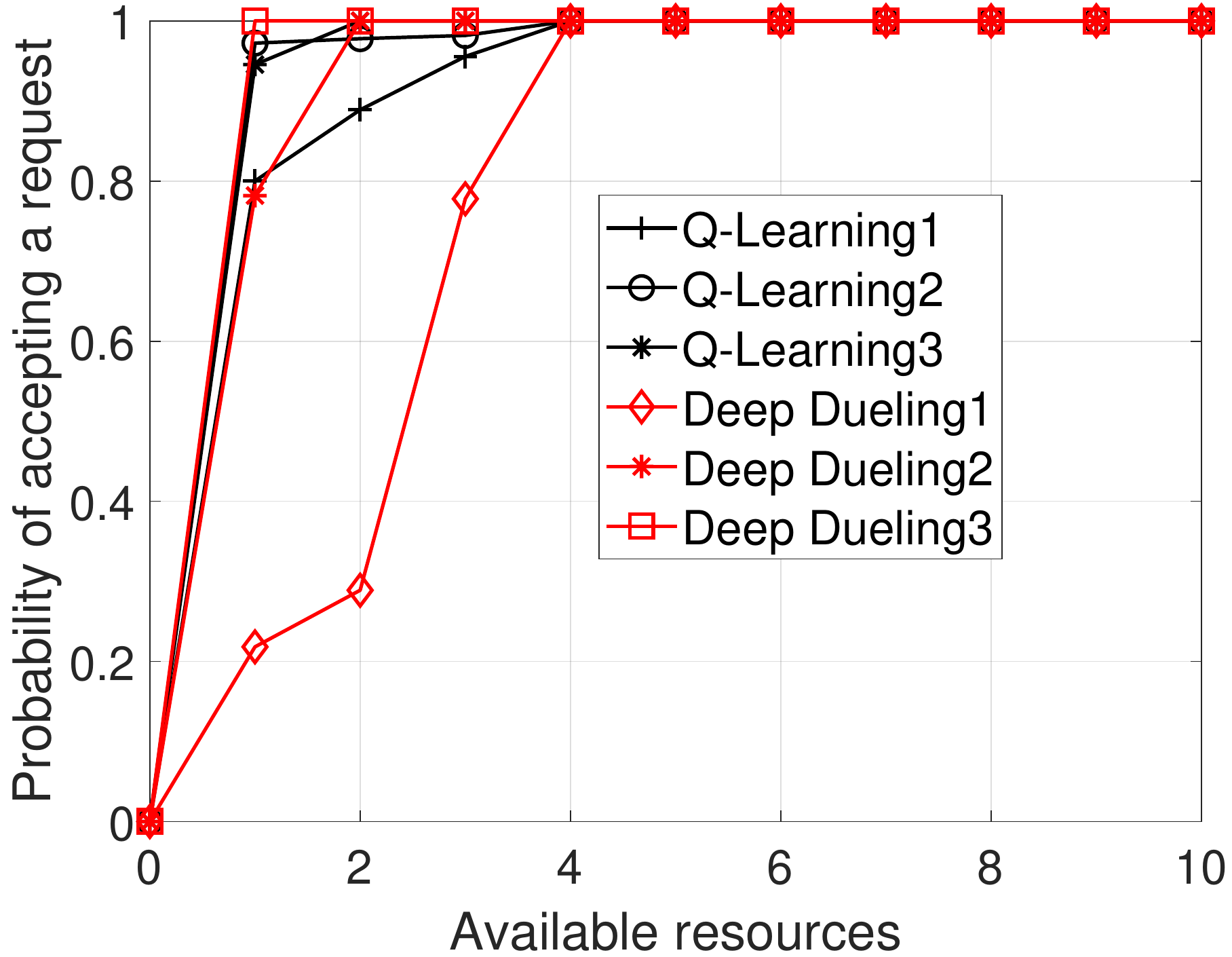}
			\caption{}
		\end{subfigure}%
		~ 
		\begin{subfigure}[b]{0.3\textwidth}
			\centering
			\includegraphics[scale=0.28]{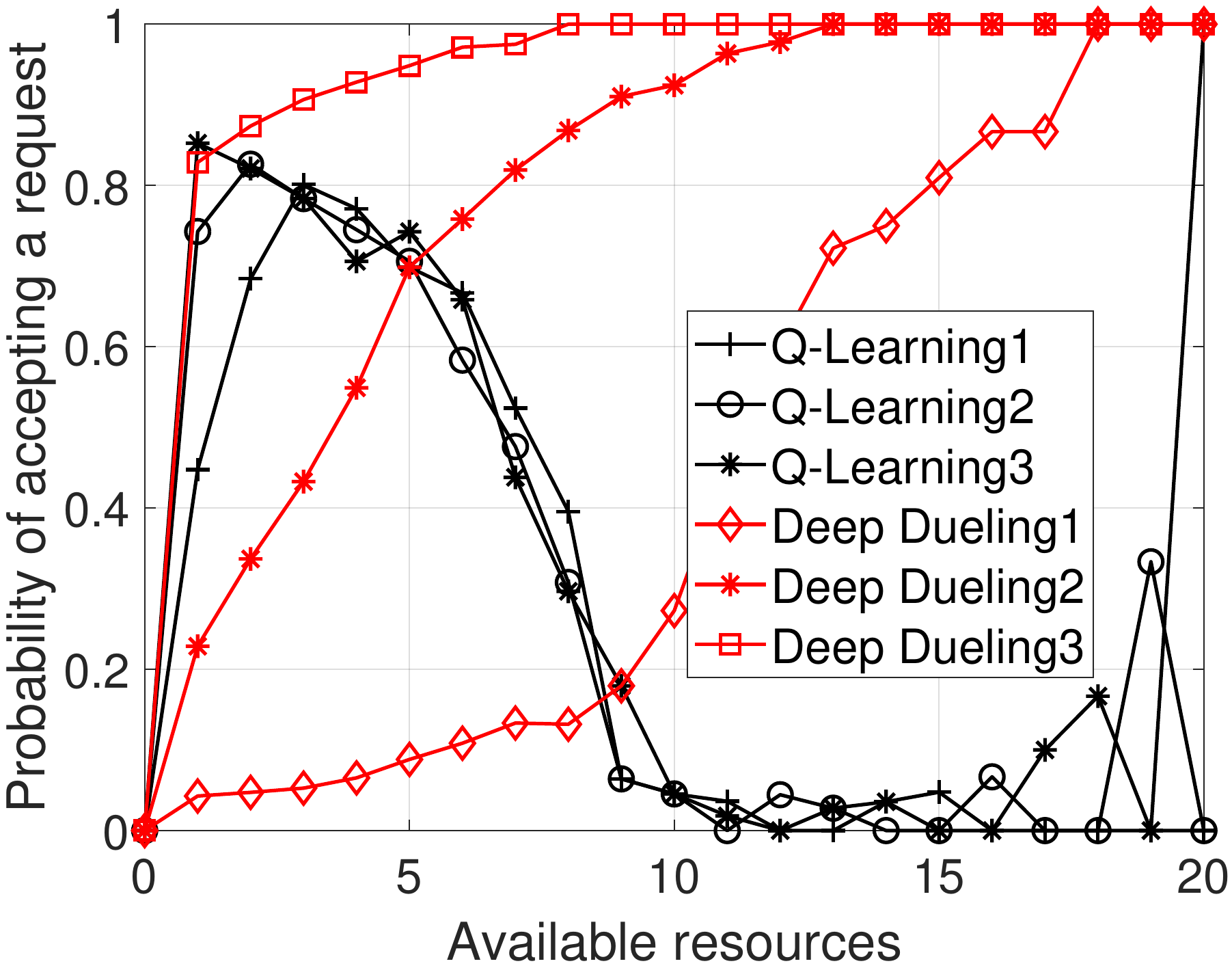}
			\caption{}
		\end{subfigure}
		\caption{The probabilities of accepting a request from classes when the maximum available resources of the system is (a) 4 times, (b) 10 times, and (c) 20 times of resources requested by a slice.} 
		\label{fig:optimal_policy}
	\end{minipage}
\end{figure*}

Similar to the case in Fig.~\ref{fig:vary_reward_running}, as shown in Fig.~\ref{fig:vary_reward_running_20}, the deep dueling and Q-learning algorithms reserve resources for slices from classes which have high immediate rewards. However, the deep dueling algorithm achieves better performance compared to the Q-learning algorithm. For example, when the immediate reward of slices from class-3 is 6, the number of requests running in the systems is about 16 requests and 11 requests for the deep dueling and the Q-learning algorithms, respectively.

In summary, in all the cases, the deep dueling algorithm always achieves the best performance in terms of the average reward and network performance.

\paragraph{Optimal Policy}
%\begin{figure*}[h]
%	\centering
%	\begin{subfigure}[b]{0.25\textwidth}
%		\centering
%		\includegraphics[scale=0.25]{Figures/Policy_4VMs.pdf}
%		\caption{}
%	\end{subfigure}%
%	~ 
%	\begin{subfigure}[b]{0.25\textwidth}
%		\centering
%		\includegraphics[scale=0.25]{Figures/Policy_10VMs}
%		\caption{}
%	\end{subfigure}%
%	~ 
%	\begin{subfigure}[b]{0.25\textwidth}
%		\centering
%		\includegraphics[scale=0.25]{Figures/Policy_20VMs}
%		\caption{}
%	\end{subfigure}
%	\caption{The probabilities of accepting a request from classes when the maximum available resources of the system is (a) 4 times, (b) 10 times, and (c) 20 times of resources requested by a slice.} 
%	\label{fig:optimal_policy}
%\end{figure*}
In Fig.~\ref{fig:optimal_policy}, we examine the optimal policy of the deep dueling and Q-learning algorithms. Specifically, we set the maximum resources of the system at 4 times, 10 times, and 20 times of resources requested by a slice and evaluate the policy of the algorithms with different available resources in the system as shown in Fig.~\ref{fig:optimal_policy}(a), Fig.~\ref{fig:optimal_policy}(b), and Fig.~\ref{fig:optimal_policy}(c), respectively. Note that the lines Q-learning\{1,2,3\} and Deep Dueling\{1,2,3\} represent the probabilities of accepting requests from class-\{1,2,3\} by using the Q-learning and deep dueling algorithms, respectively.

Clearly, in three cases, the deep dueling always obtains the best policy. In particular, it will reject almost requests from class-1 (lowest immediate reward) when there are few available resources in the system. When the available resources in the system increase, the probability of accepting a request from class-1 is also increased. Note that, when the maximum system resource capacity is large, i.e., 20 times of resources requested by a slice, the performance of the Q-learning is fluctuated as it cannot converge to the optimal policy event with $10^7$ iterations.

%==========================
%==========================
\subsubsection{Convergence of Deep Reinforcement Learning Approaches}
\begin{figure}[h]
	\centering
	\begin{subfigure}[b]{0.225\textwidth}
		\centering
		\includegraphics[scale=0.235]{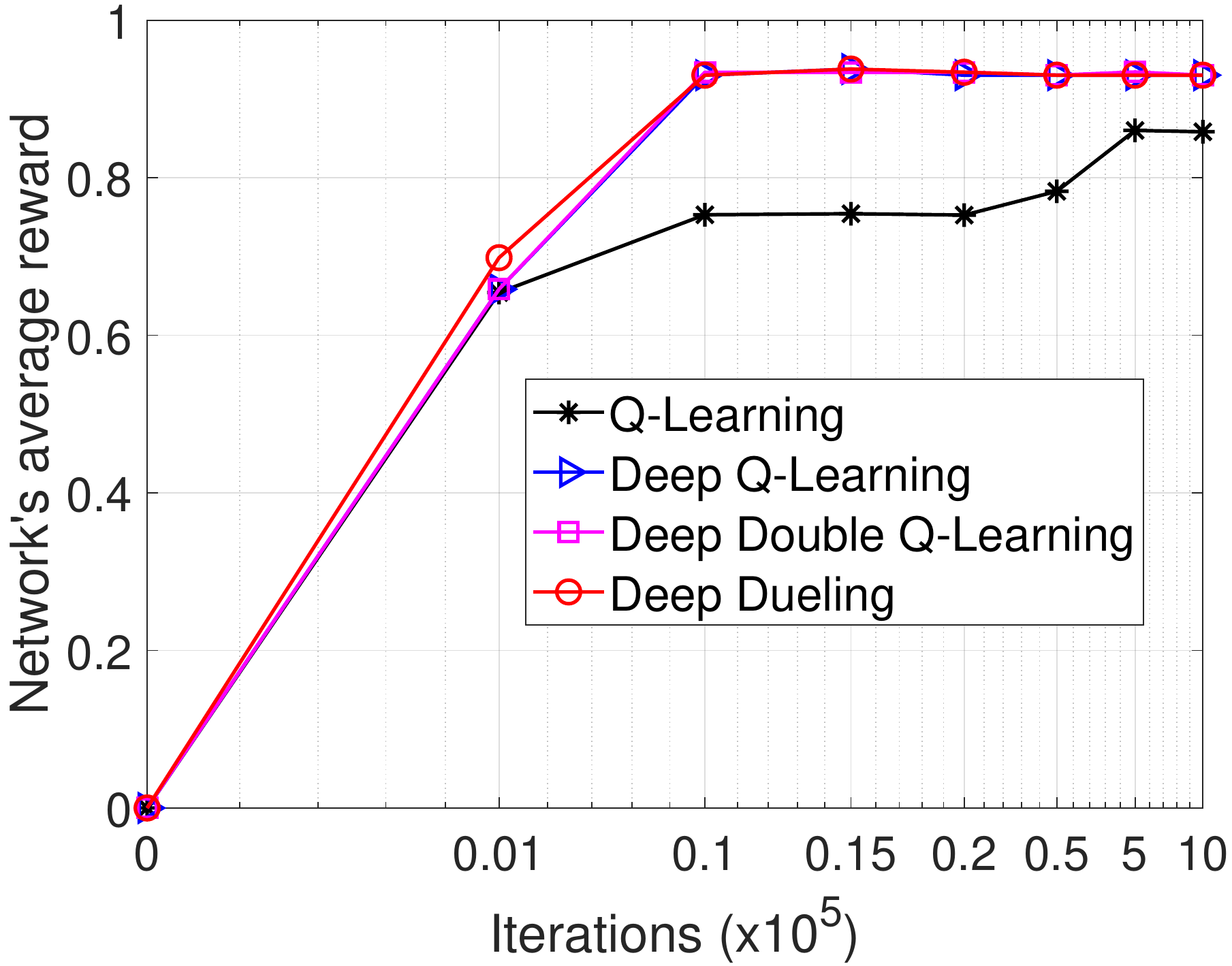}
		\caption{}
	\end{subfigure}%
	~
	\begin{subfigure}[b]{0.225\textwidth}
		\centering
		\includegraphics[scale=0.235]{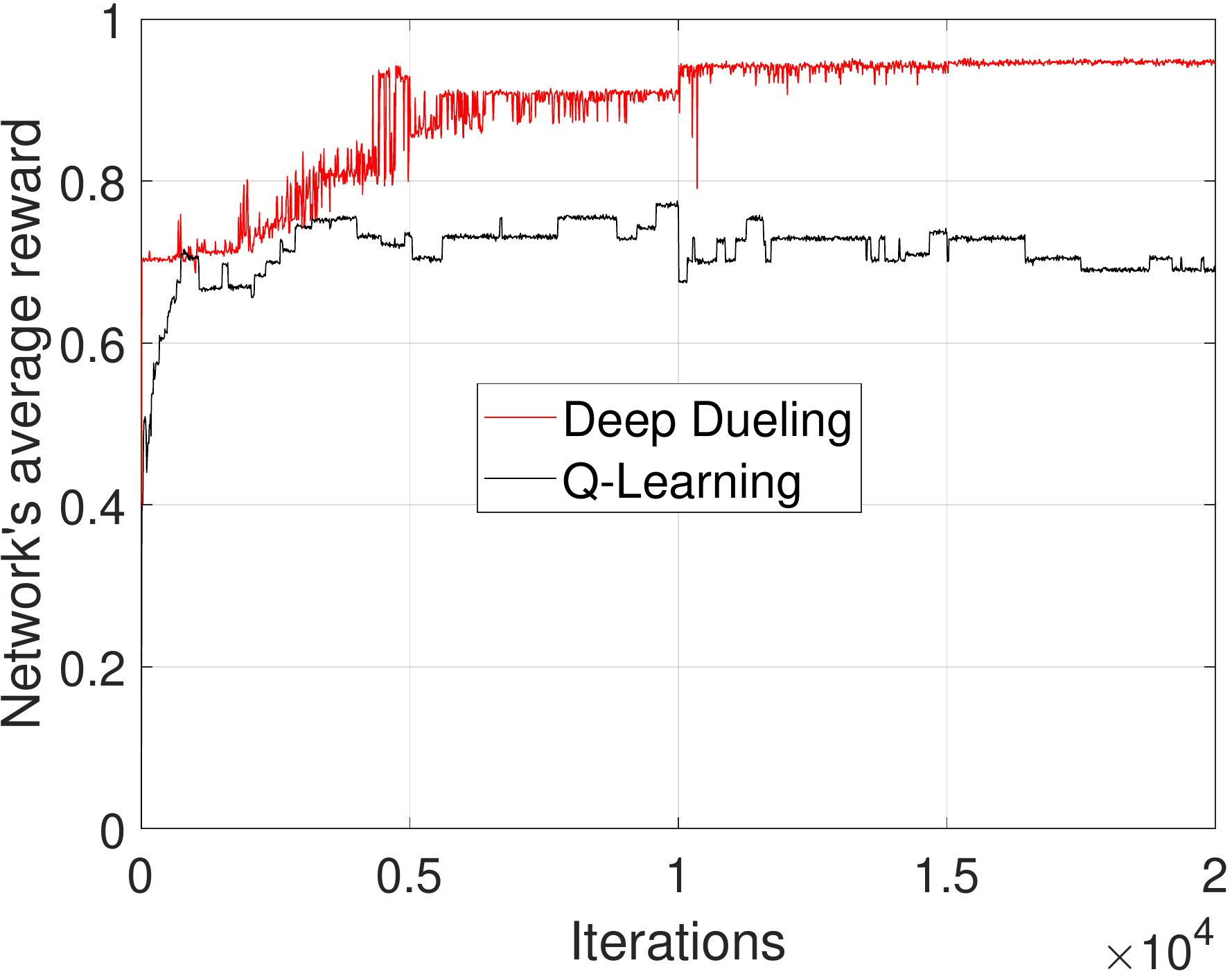}
		\caption{}
	\end{subfigure}%
	\caption{The convergence of reinforcement learning algorithms when the radio, computing, storage resources are 400 Mbps, 8 CPUs, and 4 GB, respectively with (a) $10^6$ iteration and (b) 20,000 iterations.}
	\label{fig:convergence2000}
\end{figure}

Next, we show the learning process and the convergence of the deep reinforcement learning approaches, i.e., deep Q-learning, deep double Q-learning, and deep dueling, in different scenarios. As shown in Fig.~\ref{fig:convergence2000}(a), when the maximum radio, computing, storage resources are 400 Mbps, 8 CPUs, and 4 GB, respectively, the convergence rates of the three deep Q-learning algorithms are considerably higher than that of the Q-learning algorithm. Specifically, while the deep reinforcement learning approaches converge to the optimal value within 10,000 iterations, the Q-learning need more than $10^6$ iterations to obtain the optimal policy. This is stemmed from the fact that in the system under consideration, the state space is dimensional and the system dynamically changes over time. In Fig.~\ref{fig:convergence2000}(b), we show the convergence of the Q-learning and deep dueling algorithms in the first 20,000 iterations to clearly verify this observation.  On the contrary, by implementing the neural network with fully-connected layers, the deep reinforcement algorithms can efficiently reduce the curse of dimensionality, thereby improving the convergence rate.

We continue to increase the radio, storage, computing resources to 1 Gbps, 10 GB, and 20 CPUs, respectively. The arrival rates of classes are increased by 4 times, i.e., $\lambda_1 = 48$, $\lambda_2 = 32$, and $\lambda_3 = 40$ requests/hour, while the completion rates are equal to 2 requests/hour for all classes. As shown in Fig.~\ref{fig:convergence1020}(a), the performance of the deep reinforcement algorithms is significantly higher than that of the Q-learning algorithm. It is important to note that as the state space now is more complicated than in the previous case, the deep dueling algorithm obtains the optimal policy within 15,000 iterations, while the other two deep reinforcement learning approaches require more time to converge to the optimal policy.
\begin{figure}[h]
	\centering
	\begin{subfigure}[b]{0.23\textwidth}
		\centering
		\includegraphics[scale=0.23]{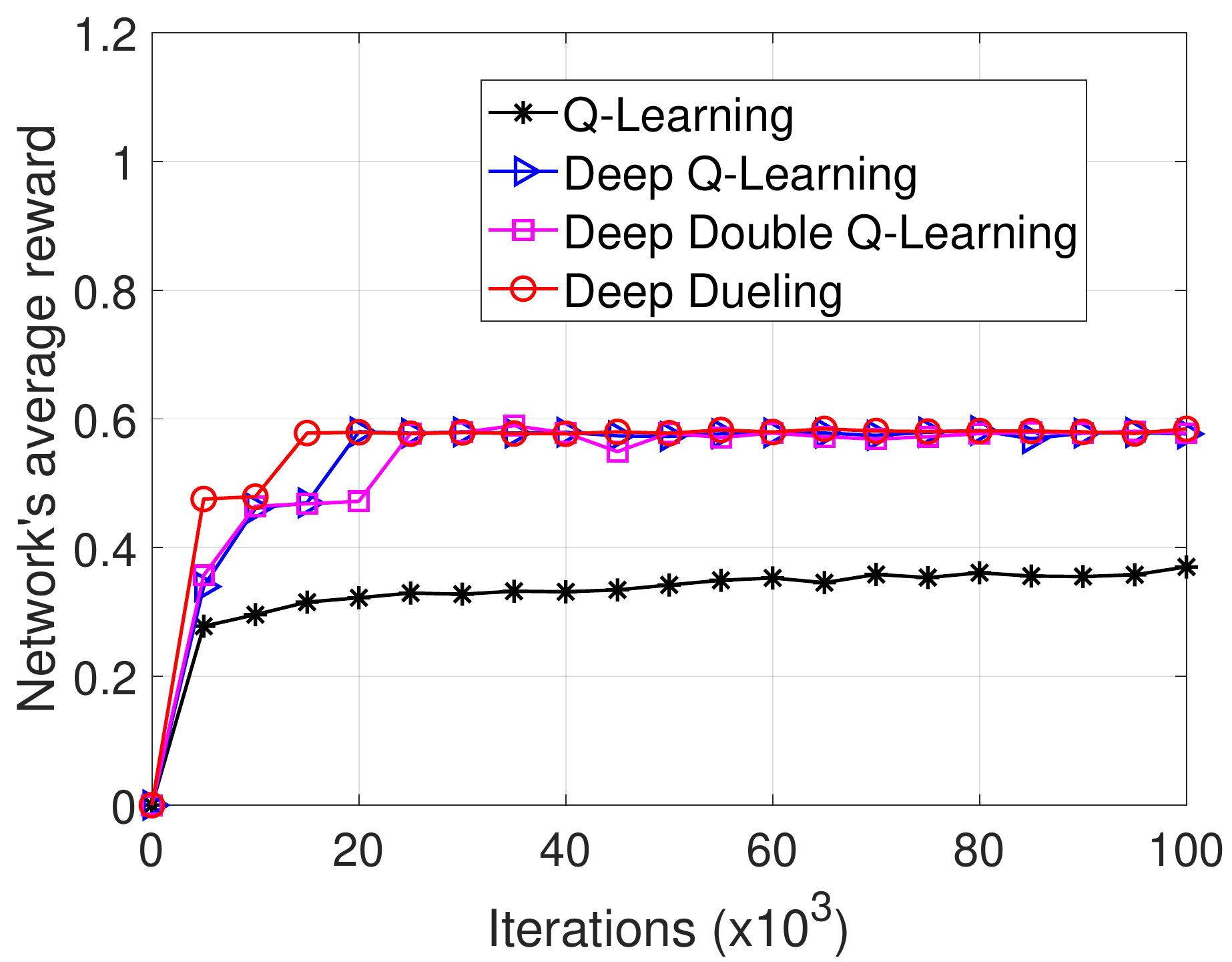}
		\caption{}
	\end{subfigure}%
	~ 
	\begin{subfigure}[b]{0.23\textwidth}
		\centering
		\includegraphics[scale=0.23]{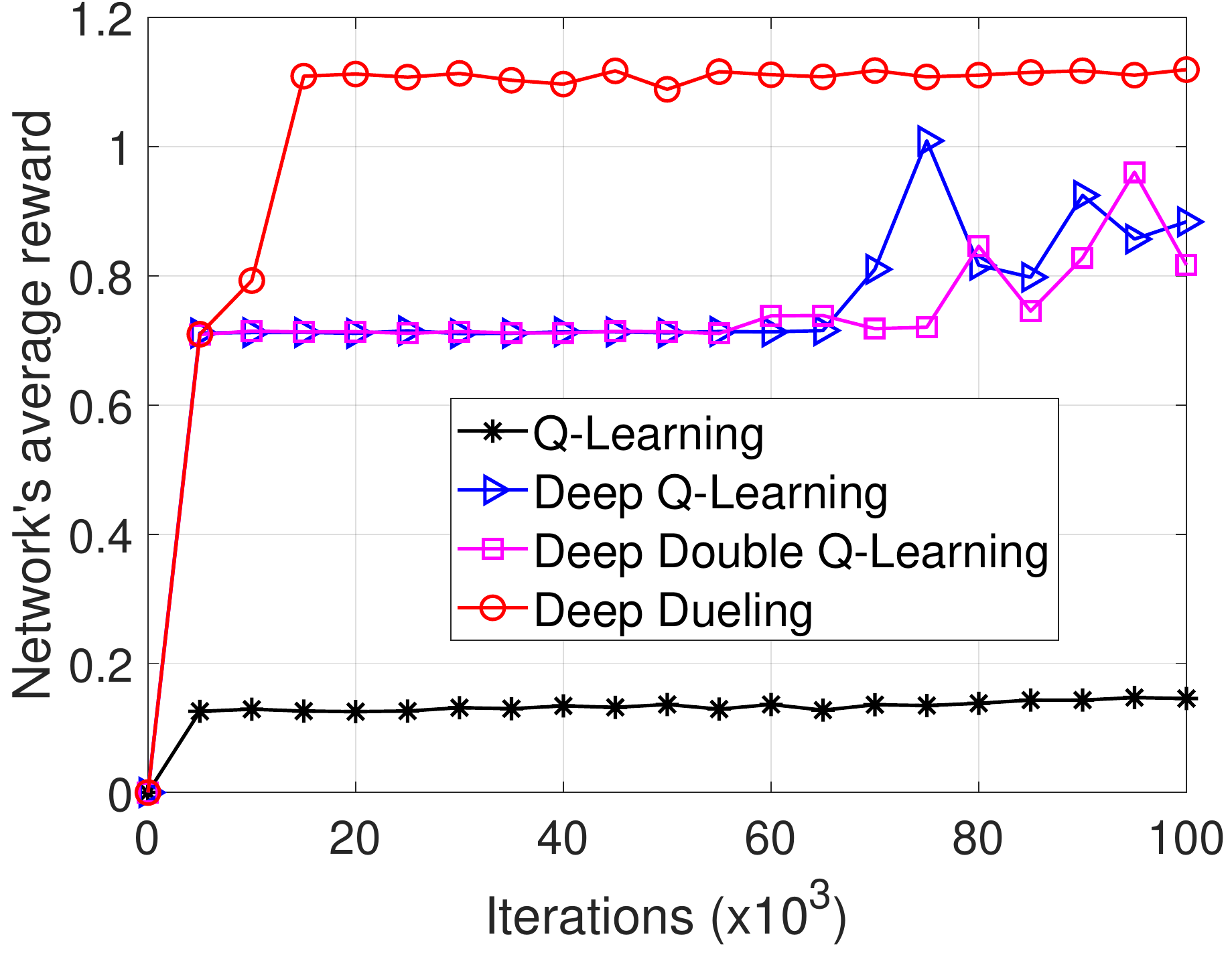}
		\caption{}
	\end{subfigure}%
	\caption{The convergence of reinforcement learning algorithms when (a) the radio, computing, storage resources are 1 Gbps, 20 CPUs, and 10 GB, respectively and (b) the radio, computing, storage resources are 2 Gbps, 20 GB, and 40 CPUs, respectively.}
	\label{fig:convergence1020}
\end{figure}

We keep increasing the radio, storage, computing resources to 2 Gbps, 20 GB, and 40 CPUs, respectively and observe the convergence rate of the deep reinforcement algorithms as shown in Fig.~\ref{fig:convergence1020}(b). Clearly, as now the system is very complicated, the deep dueling can achieve the optimal policy within 20,000 iterations while the deep Q-learning and deep double Q-learning algorithms cannot converge to the optimal policy after 100,000 iterations. This is due to the fact that by decoupling the neural network into two streams, the deep dueling algorithm can significantly reduce the overestimation of the optimizer, i.e., stochastic gradient descent. 

\begin{figure}[h]
	\centering
	\includegraphics[scale=0.33]{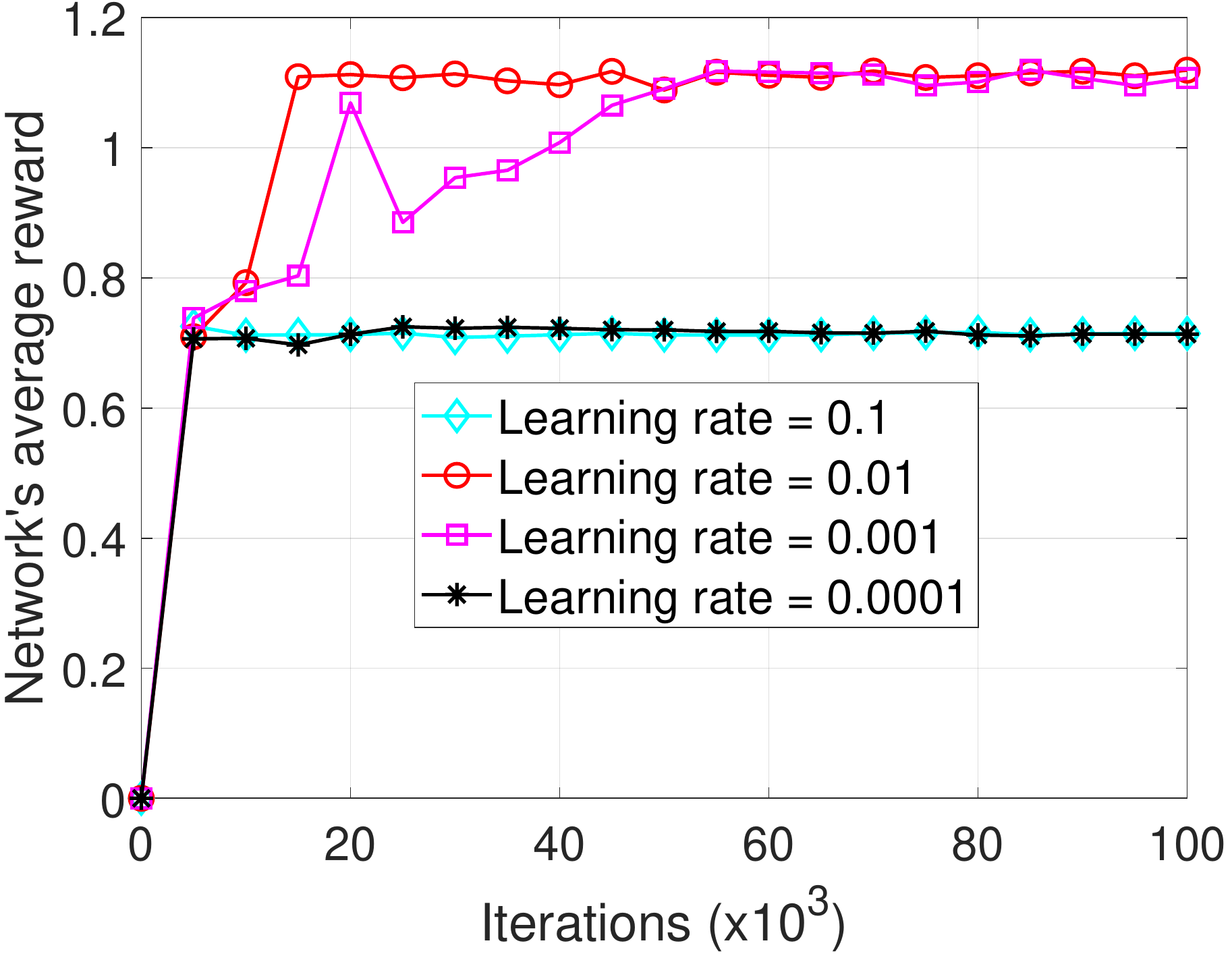}
	\caption{The performance of deep dueling algorithm with different learning rates.}
	\label{Fig.vary_learning_rate}
\end{figure}
Next, we show the effects of the learning rate on the performance of the deep dueling algorithm. The learning rate is the most critical hyper-parameters to tune for training deep neural networks. If the learning rate is too slow, the training process is more reliable but requires a long time to converge to the optimal policy. In contrast, if the learning rate is too high, the algorithm may not converge to the optimal policy or even diverge. This is stemmed from the fact that the deep dueling algorithm uses the gradient descent method. If the learning rate is too large, gradient descent may overshoot the optimal point, and thus resulting in poor performance. This observation is proved by simulation results as shown in Fig.~\ref{Fig.vary_learning_rate}. Specifically, with the learning rate is 0.01, the deep dueling algorithm achieves the best performance in terms of the average reward and the convergence rate compared to other learning rates.

%%%%%%%%%%%%%%%%%%%%%%%%%%%%%%%%%%%%%%%%%%%%%%%%	
%%%%%%%%%%%%%%%%%%%%%%%%%%%%%%%%%%%%%%%%%%%%%%%%

%--------------------------------------------------------------------------------------------------------------------------
%--------------------------------------------------------------------------------------------------------------------------	
\section{Conclusion}
\label{sec:conclusion}
In this paper, we have developed the optimal and fast network resource management framework which allows the network provider to jointly allocate multiple combinatorial resources (i.e., computing, storage, and radio) to different slice requests in a real-time manner. To deal with the dynamic and uncertainty of slice requests, we have adopted the semi-Markov decision process. Then, the reinforcement learning algorithms, i.e., Q-learning, deep Q-learning, deep double Q-learning, and deep dueling, have been employed to maximize the long-term average reward for the network provider. The key idea of the deep dueling is using two streams of fully connected hidden layers to concurrently train the value and advantage functions, thereby improving the training process and achieving the outstanding performance for the system. Extensive simulations have shown that the proposed framework using deep dueling can yield up to 40\% higher long-term average reward with few thousand times faster compared with those of other network slicing approaches. Future works comprise considering the connectivity resources and the existence of multiple data centers in complex network slicing models by accommodating more states to the system state space. The performance of the proposed solution will be evaluated in terms of complexity and scalability. Moreover, the convergence rate and stability of the deep dueling algorithm will be improved by using the state of the art deep neural networks.

%--------------------------------------------------------------------------------------------------------------------------
%--------------------------------------------------------------------------------------------------------------------------		
\appendices

%%%%%%%%%%%%%%%%%%%%%%%%%%%%%%%%%%%%%%%%%%%%%%%%
%%%%%%%%%%%%%%%%%%%%%%%%%%%%%%%%%%%%%%%%%%%%%%%%
\section{The proof of Theorem~\ref{theorem_equivalent}}
\label{appendix:pro_equivalent}
For any $t \ge 0$, define the matrix $\mathbf{P}(t)$ by $\mathbf{P}(t) = (p_{\mathbf{s},\mathbf{s'}}(t)), \forall \mathbf{s}, \mathbf{s'} \in \mathcal{S}$. Denote by $\mathbf{Q}$, the matrix $\mathbf{Q}=q_{\mathbf{s},\mathbf{s'}}, \forall \mathbf{s}, \mathbf{s'} \in \mathcal{S}$, where the diagonal elements $q_{\mathbf{s},\mathbf{s}}$ are defined by:
\begin{equation}
q_{\mathbf{s},\mathbf{s}} = -z_{\mathbf{s}}.
\end{equation}

After that Kolmogoroff's forward differential equations can be written as $\mathbf{P}'(t)= \mathbf{P}(t)\mathbf{Q}$ for any $t\ge 0$. Hence, the solution of this system of differential equations is given by:
\begin{equation}
\mathbf{P}(t) = e^{t\mathbf{Q}} = \sum_{n=0}^{\infty}\frac{t^n}{n!}\mathbf{Q}^n, t \geq 0.
\end{equation}
The matrix $\overline{\mathbf{P}} = \overline{p}_{\mathbf{s}, \mathbf{s'}}, \forall \mathbf{s}, \mathbf{s'} \in \mathcal{S}$ can be reformulated as $\overline{\mathbf{P}} = \mathbf{Q}/z + \mathbf{I}$, where $\mathbf{I}$ is the identity matrix. Therefor, we have
\begin{equation}
\begin{aligned}
\mathbf{P}(t) = e^{t\mathbf{Q}} = e^{zt(\overline{\mathbf{P}}- \mathbf{I})} &= e^{zt\overline{\mathbf{P}}}e^{-zt\mathbf{I}}=e^{-zt}e^{zt\overline{\mathbf{P}}}\\
&= \sum_{n=0}^{\infty}e^{-zt}\frac{{(zt)}^n}{n!}\overline{\mathbf{P}}^n.
\end{aligned}
\end{equation}
Based on conditioning on the number of Poisson events up to time $t$ in the $\{\overline{X}(t)\}$ process, we have
\begin{equation}
P\{\overline{X}(t)=\mathbf{s'}|\overline{X}(0)=\mathbf{s}\}=\sum_{n=0}^{\infty}e^{-zt}\frac{{(zt)}^n}{n!}\overline{p}_{\mathbf{s},\mathbf{s'}}^{(n)},
\end{equation}
where $\overline{p}_{\mathbf{s},\mathbf{s'}}^{(n)}$ is the $n$-step transition probability of the discrete-time Markov chain $\overline{X}_n$. By recalling the Corollary~\ref{cor}, the proof is completed.
\section{The proof of Lemma~\ref{Lem:exist}}
\label{app:reward}
Let $\{A_n: n \geq 0\}$ be a sequence of matrices. We have $\lim\limits_{n \rightarrow \infty} A_n=A$ if $\lim\limits_{n \rightarrow \infty} A_n(\mathbf{s'}|\mathbf{s})=(\mathbf{s'}|\mathbf{s})$ for each $(\mathbf{s},\mathbf{s'}) \in \mathcal{S} \times \mathcal{S}$. We now consider the Cesaro limit which is defined as follows. We say that $A$ is the Cesaro limit (of order one) of $\{A_n: n \geq 0\}$ if
\begin{equation}
\lim\limits_{n \rightarrow \infty}\frac{1}{N}\sum_{n=0}^{N-1}A_n=A,
\end{equation}
and write
\begin{equation}
	C-\lim_{N \rightarrow \infty}A_N=A
\end{equation}
to distinguish this as a Cesaro limit.
We then define the limiting matrix $\overline{P}$ by
\begin{equation}
\overline{P}=C-\lim_{N \rightarrow \infty}P^N.
\end{equation}
In component notation, where $\overline{p}(\mathbf{s'}|\mathbf{s})$ denotes the $(\mathbf{s'}|\mathbf{s})$-th element of $\overline{P}$, this means that, for each $\mathbf{s}$ and $\mathbf{s'}$, we have
\begin{equation}
\overline{p}(\mathbf{s'}|\mathbf{s}) = \lim_{N \rightarrow \infty}\frac{1}{N}\sum_{n=1}^{N}p^{n-1}(\mathbf{s'}|\mathbf{s}),
\end{equation}
where $p^{n-1}$ denotes a component of $P^{n-1}$ and $p^0(s'|s)$ is a component of an $\mathcal{S}\times\mathcal{S}$ identity matrix. As $P$ is aperiodic, $\lim_{N \rightarrow \infty}$ exists and equals to $\overline{P}$.
%%%%%%%%%%%%%%%%%%%%%%%%%%%%%%%%%%%%%%%%%%%%%%%%	
%%%%%%%%%%%%%%%%%%%%%%%%%%%%%%%%%%%%%%%%%%%%%%%%
%\section{The proof of Proposition~\ref{prop2}}
%\label{appendix:prop2}

%--------------------------------------------------------------------------------------------------------------------------
%--------------------------------------------------------------------------------------------------------------------------	

\end{document}